\documentclass[11pt]{article}

\usepackage{amsfonts}
\usepackage{amssymb}
\usepackage{amstext}
\usepackage{amsmath}
\usepackage{mathtools}

\usepackage{amsthm} % NEED TO REMOVE amsthm in ACM EC Format

\usepackage{caption}
\usepackage{subcaption}

\usepackage{color}
\usepackage{nameref}
\definecolor{ForestGreen}{rgb}{0.1333,0.5451,0.1333}
\definecolor{DarkRed}{rgb}{0.8,0,0}
\definecolor{Red}{rgb}{0.9,0,0}
\usepackage[linktocpage=true,
	pagebackref=true,colorlinks,
	linkcolor=DarkRed,citecolor=ForestGreen,
	bookmarks,bookmarksopen,bookmarksnumbered]
	{hyperref}
\usepackage{cleveref}

\usepackage{thmtools,thm-restate} % See section 1.4 of the pdf above

\usepackage[numbers,sort&compress]{natbib}

\usepackage{graphicx}
\usepackage{graphics}
\usepackage{colordvi}
\usepackage{xspace}
\usepackage{algorithm}
\usepackage{algorithmicx}
\usepackage{url}
\usepackage{enumitem}

        {\hspace*{\fill}$\Box$\par\vspace{4mm}}

\usepackage[left=1in,top=1in,right=1in,bottom=1in]{geometry} % Does NOT work right now
\usepackage{booktabs}
\usepackage{threeparttable}
\usepackage{times} % Even more aggressive saver

\makeatletter
\renewcommand{\paragraph}{%
  \@startsection{paragraph}{4}%
  {\z@}{1ex \@plus 1ex \@minus .2ex}{-1em}%
  {\normalfont\normalsize\bfseries}%
}
\makeatother

\makeatletter
\def\thmt@refnamewithcomma #1#2#3,#4,#5\@nil{%
  \@xa\def\csname\thmt@envname #1utorefname\endcsname{#3}%
  \ifcsname #2refname\endcsname
    \csname #2refname\expandafter\endcsname\expandafter{\thmt@envname}{#3}{#4}%
  \fi
}
\makeatother

\declaretheorem[numberwithin=section,refname={Theorem,Theorems},Refname={Theorem,Theorems}]{theorem}
\declaretheorem[numberlike=theorem,refname={Lemma,Lemmas},Refname={Lemma,Lemmas}]{lemma}

\declaretheorem[numberlike=theorem,refname={Corollary,Corollaries},Refname={Corollary,Corollaries}]{corollary}

\declaretheorem[numberlike=theorem,refname={observation,observations},Refname={Observation,Observations}]{observation}

\declaretheorem[numberlike=theorem]{definition}

\newtheorem{invariant}[theorem]{Invariant}
\newtheorem{Claim}[theorem]{Claim}
\renewcommand{\phi}{\varphi}

\newcommand{\N}{\mathcal{N}}
\newcommand{\B}{\mathcal{B}}

\newcommand{\eps}{\epsilon}

\newcommand{\dd}{D}

\newcommand{\E}{\mathcal{E}}

\ifdefined\ShowComment

\def\danupon#1{\marginpar{$\leftarrow$\fbox{D}}\footnote{$\Rightarrow$~{\sf #1 --Danupon}}}
\def\babis#1{\marginpar{$\leftarrow$\fbox{B}}\footnote{$\Rightarrow$~{\sf #1 --Babis}}}
\def\sayan#1{\marginpar{$\leftarrow$\fbox{S}}\footnote{$\Rightarrow$~{\sf #1 --Sayan}}}
\def\monika#1{\marginpar{$\leftarrow$\fbox{M}}\footnote{$\Rightarrow$~{\sf #1 --Monika}}}

\else

\def\danupon#1{}
\def\babis#1{}
\def\sayan#1{}
\def\monika#1{}

\fi

\newboolean{short}
\setboolean{short}{true} % set to true if the paper is the short version

\newcommand{\shortOnly}[1]{\ifthenelse{\boolean{short}}{#1}{}}
\newcommand{\longOnly}[1]{\ifthenelse{\boolean{short}}{}{#1}}

\title{Design of Dynamic Algorithms via Primal-Dual Method\thanks{A preliminary version of this paper appeared in ICALP' 2015 (Track A).}}

\author{Sayan Bhattacharya\thanks{The Institute of Mathematical Sciences, Chennai, India. Email: bsayan@imsc.res.in}
	\and Monika Henzinger\thanks{Faculty of Computer Science, University of Vienna, Austria. Email: monika.henzinger@univie.ac.at. The research leading to this work has received funding from the European Union's Seventh Framework Programme (FP7/2007-2013) under grant agreement number 317532 and from the European Research Council under the European Union's Seventh Framework Programme (FP7/2007-2013)/ERC grant agreement number 340506.}       
    \and Giuseppe F. Italiano\thanks{Universit\`a di Roma ``Tor Vergata'', Italy. E-mail: giuseppe.italiano@uniroma2.it. Partially supported by
MIUR, the Italian Ministry
of Education, University and Research, under Project AMANDA
(Algorithmics for MAssive and Networked DAta).}   
}

\begin{document}

%\maketitle % (NEED TO BE AFTER Abstract IN ACM EC Format)

%\begin{titlepage}
\maketitle
\pagenumbering{roman}

\setcounter{tocdepth}{3}
%\tableofcontents

\pagenumbering{arabic}

\begin{abstract}
We develop a dynamic version of the primal-dual method for optimization problems, and  apply it to obtain the following results. (1) For the dynamic set-cover problem, we maintain an $O(f^2)$-approximately optimal solution in $O(f \cdot  \log (m+n))$ amortized update time, where $f$ is the maximum ``frequency'' of an element, $n$ is the number of sets, and $m$ is the maximum number of elements in the universe at any point in time. (2) For the dynamic  $b$-matching problem, we maintain an $O(1)$-approximately optimal solution in $O(\log^3 n)$ amortized update time, where $n$ is the number of nodes in the  graph. 
\end{abstract}

\section{Introduction}

%\paragraph{Importance of Primal-Dual} 

The primal-dual method lies at the heart of the design
of algorithms for combinatorial optimization problems.
The basic idea, contained in the ``Hungarian Method''~\cite{Kuh55}, was extended and formalized by Dantzig et al.~\cite{DFF56} as a general framework for linear programming, and thus it became applicable to a large variety of problems. 
%\margin{PINO: if we need space, we can remove this pargraph with its 3 references. This would save about 10 lines.} 
%Indeed,  most of the fundamental algorithms in combinatorial optimization, including Ford and Fulkerson's network flow algorithm~\cite{FF56}, Dijkstra's shortest path algorithm~\cite{Dij59}  and Edmonds' maximum matching algorithm~\cite{Edm65}, either use the primal-dual method directly or can be reinterpreted within its framework.
%
Few decades later,  Bar-Yehuda et al.~\cite{BYE81} were the first to apply the primal-dual method to the design of approximation algorithms. Subsequently, 
this paradigm was  applied to obtain approximation algorithms for a wide collection of  NP-hard problems~\cite{GoemansW92,GW97}. 
When the primal-dual method is applied to approximation algorithms, an approximate solution to the problem and
a feasible solution to the dual of an LP relaxation are constructed simultaneously, and the
performance guarantee is proved by comparing the values of both solutions. 
The primal-dual method was also extended to
 online problems~\cite{BuchbinderN09}. Here,  
  the input
is revealed only in parts, and an online algorithm is required to respond to
each new input upon its arrival (without being able to see the future). 
The algorithm's performance is compared against the  benchmark of an optimal omniscient algorithm that can view the entire input sequence in advance. 
 
In this paper, we focus on dynamic algorithms for optimization problems. In the dynamic setting, the input of a problem is being changed via a sequence of updates, and after each update one is interested in maintaining the solution to the problem much faster than recomputing it from scratch. We remark that the dynamic and the online setting are completely different: in the dynamic scenario one is concerned more with guaranteeing fast (worst-case or amortized) update times rather than comparing the algorithms' performance against optimal offline algorithms. 
%
%Up to now, the primal-dual method has been mainly applied to static problems only, both in the case of exact polynomial-time algorithms and in the case of approximate algorithms for NP-hard problems.
As a main contribution of this paper, we develop a dynamic version of the primal-dual method, thus opening up a completely new area of application of the primal-dual paradigm to the design of dynamic algorithms. 
With some careful insights, our recent algorithms for dynamic matching and dynamic vertex cover~\cite{BHI15} can be reinterpreted in this new framework. In this paper, 
we show how to apply the new dynamic primal-dual framework to the design of 
two other optimization problems: the dynamic  set-cover  and the dynamic $b$-matching. Before proceeding any further, we formally define these problems. 

%\vspace{-.1cm}

\begin{definition}[Set-Cover]
\label{main:def:set-cover}
We are given a universe $\mathcal{U}$ of at most $m$ elements, and a collection $\mathcal{S}$ of $n$ sets $S \subseteq \mathcal{U}$. Each set $S \in \mathcal{S}$ has a (polynomially bounded by $n$) ``cost'' $c_S > 0$. The goal is to select a subset $\mathcal{S}' \subseteq \mathcal{S}$ such that each element in $\mathcal{U}$ is covered by some set $S \in \mathcal{S}'$ and the total cost $\sum_{S \in \mathcal{S}'} c(S)$ is minimized. 
\end{definition}

%\vspace{-.4cm}

\begin{definition}[Dynamic Set-Cover]
\label{main:def:dynamic:set-cover}
Consider a dynamic version of the problem specified in Definition~\ref{main:def:set-cover},  where the collection $\mathcal{S}$, the costs $\{c_S\}, S \in \mathcal{S}$,  the upper bound $f$ on the maximum frequency  $\max_{u \in \mathcal{U}} |\{ S \in \mathcal{S} : u \in S\}|$, and the upper bound $m$ on the maximum size of the universe $\mathcal{U}$ remain fixed. The universe $\mathcal{U}$, on the other hand, keeps changing dynamically.   In the beginning, we have $\mathcal{U} = \emptyset$. At each time-step, either an element $u$ is inserted into the universe $\mathcal{U}$ and we get to know which sets in $\mathcal{S}$ contain $u$,  or some element is deleted from the universe. The goal is to maintain an approximately optimal solution to the set-cover problem in this dynamic setting.
\end{definition}

%\vspace{-.4cm}

\begin{definition}[$b$-Matching]
\label{main:def:b-matching}
We are given an input graph $G = (V, E)$ with $|V| = n$ nodes, where each node $v \in V$ has a capacity $c_v \in \{1, \ldots, n\}$. 
A $b$-matching is a  subset $E' \subseteq E$ of edges  such that each node $v$ has at most $c_v$ edges incident to it in $E'$. The goal is to select the $b$-matching of maximum cardinality. 
\end{definition}

%\vspace{-.4cm}

\begin{definition}[Dynamic $b$-Matching]
\label{main:def:dynamic:b-matching}
Consider a dynamic version of the problem specified in Definition~\ref{main:def:b-matching},  where the node set $V$ and the capacities $\{c_v\}, v \in V$ remain fixed. The edge set $E$, on the other hand, keeps changing dynamically.   In the beginning, we have $E = \emptyset$. At each time-step, either a new edge is inserted into the graph or some existing edge is deleted from the graph. The goal is to maintain an approximately optimal solution to the $b$-matching problem in this dynamic setting.
\end{definition}

%\vspace{-.1cm}

As stated in~\cite{BuchbinderN09,Vazirani01}, 
the set-cover problem has played a pivotal role both for approximation and for online algorithms, and thus it seems a natural problem to consider in our dynamic setting. Our definition of dynamic set-cover  is inspired by the standard formulation of the online set-cover problem~\cite{BuchbinderN09}, where the elements arrive  online.  There exists algorithms for online set cover that achieve a competitive ratio of $O(\log n \log m)$~\cite{BuchbinderN09}, and it is also known that this bound is asymptotically tight~\cite{Korman}.

%\smallskip
%\noindent 
\paragraph{Our Techniques.}
Roughly speaking, our dynamic version of the primal-dual method works as follows. We start with a feasible primal solution  and an infeasible dual solution for the problem at hand. Next, we consider the following process:
gradually increase all the primal variables at the same rate, and whenever a primal constraint becomes tight,  stop the growth of all the primal variables involved in that constraint, and update accordingly the corresponding dual variable. This primal growth process is used to define a suitable data structure based on a hierarchical partition. A level in this partition is a set of the dual variables whose corresponding primal constraints became (approximately) tight at the same time-instant. To solve the dynamic problem,  we maintain the data structure, the hierarchical partition and the corresponding primal-dual solution dynamically using a simple greedy procedure. 
This is sufficient for solving the dynamic set-cover problem. For the dynamic $b$-matching problem, we need some additional ideas. We first
get a fractional solution to the problem using the previous technique. To obtain an integral solution,
 we perform randomized rounding on the fractional solution
 in a dynamic setting. This is done
 by sampling the edges with probabilities that are determined by the fractional solution.

%\smallskip
%\noindent 
\paragraph{Our Results.}
Our new dynamic primal-dual framework yields efficient dynamic algorithms for both the dynamic set-cover problem and the dynamic $b$-matching problem.
In particular, 
for the dynamic set-cover problem
we  maintain a $O(f^2)$-approximately optimal solution in $O(f \cdot  \log (m+n))$ amortized update time (see Theorem~\ref{main:cor:set-cover} in Section~\ref{sec:set-cover}). On the other hand,  
for the dynamic $b$-matching problem, we maintain a $O(1)$-approximation in $O(\log^3 n)$ amortized time per update (see Theorem~\ref{th:sample:main} in Section~\ref{sec:bmatching}). Further, we can show that an edge insertion/deletion in the input graph, on average, leads to  $O(\log^2 n)$  changes in the set of matched edges maintained by our algorithm.

%\smallskip
%\noindent 
\paragraph{Related Work.}
%First paragraph: we say that dynamic algorithms is a well established research area, and give some references on the various kinds of problems that have been solved in the dynamic setting. 
The design of dynamic algorithms is one of the classic areas in theoretical computer science with a countless
number of applications. 
%The goal of a dynamic algorithm is to update efficiently the solution of a problem in the presence of dynamically changing input data, rather than having to recompute it from scratch after each update.
Dynamic graph algorithms have received special attention, and there have been many efficient algorithms for several dynamic graph problems, including dynamic connectivity, minimum spanning trees, transitive closure, shortest paths and matching problems (see, e.g., the survey in~\cite{EGI09}).
%
%Second paragraph: we describe previous results on dynamic matching and say that none of these results extend to the b-matching problem.
%
The $b$-matching problem contains as a special case matching problems, for which many dynamic algorithms are known~\cite{BaswanaGS11,BHI15,GuptaP13,NeimanS13,OnakR10}. 
Unfortunately, none of the results on dynamic matching extends to the dynamic $b$-matching problem. To the best of our knowledge, no previous result was known for dynamic set-cover problem. 

In the static setting, a simple greedy algorithm for the set-cover problem gives $O(\log n)$ approximation~\cite{johnson}, whereas a primal-dual algorithm gives $f$-approximation~\cite{BYE81}. Both the algorithms run in $O(f \cdot (m+n))$-time. On the other hand, there exists some constant $c > 0$ such that obtaining a $c \log n$-approximation to the set cover problem in polynomial time will imply $P = NP$~\cite{Feige-setcover}. Similarly, under the Unique-Games conjecture, one cannot obtain a better than $f$-approximation to the set cover problem in polynomial time~\cite{unique-games}. 

For the maximum $b$-matching problem, the best known exact algorithm  runs in $O(m n \log n)$-time~\cite{Gabow} in the static setting, where $n$ (resp. $m$) is the number of nodes (resp. edges) in the graph. Very recently, Ahn and Guha~\cite{bmatching} presented another static algorithm that runs in $O(m  \cdot \text{poly} (\delta^{-1}, \log n))$-time and returns a $(1+\delta)$-approximation for maximum $b$-matching, for any $\delta > 0$.

\paragraph{Roadmap for the rest of the paper.} We first define a problem called ``fractional hypergraph $b$-matching'' (see Definitions~\ref{main:def:fractional:bmatching} and~\ref{main:def:dynamic:fractional:bmatching}). In Section~\ref{main:sec:set-cover}, we show how to maintain a fractional hypergraph $b$-matching in a dynamic setting. In Section~\ref{sec:set-cover}, we use our result from Section~\ref{main:sec:set-cover} to design a dynamic algorithm for set cover. Finally, in Section~\ref{sec:bmatching} we present our result for dynamic $b$-matching.

\begin{definition}[Fractional Hypergraph $b$-Matching]
\label{main:def:fractional:bmatching}
We are given an input hypergraph $G = (V, E)$ with $|V| = n$ nodes and {\em at most} $m \geq |E|$ edges. Let $\mathcal{E}_v \subseteq E$ denote the set of edges incident upon a node $v \in V$, and let $\mathcal{V}_e = \{ v \in V : e \in \mathcal{E}_v\}$ denote the set of nodes an edge $e \in E$ is incident upon. Let $c_v > 0$ denote the ``capacity'' of a node $v \in V$, and let $\mu \geq 1$ denote the ``multiplicity'' of an edge. We assume that the $\mu$ and the $c_v$ values are polynomially bounded by $n$. Our goal is to assign a ``weight'' $x(e) \in [0, \mu]$ to each edge $e \in E$ in such a way that (a) $\sum_{e \in \mathcal{E}_v} x(e) \leq c_v$ for all nodes $v \in V$, and (b) the sum of the weights of all the edges is maximized. 
\end{definition}

\begin{definition}[Dynamic Fractional Hypergraph $b$-Matching]
\label{main:def:dynamic:fractional:bmatching}
Consider a dynamic version of the problem specified in Definition~\ref{main:def:fractional:bmatching},  where the node-set $V$, the capacities $\{c_v\}, v \in V$,  the upper bound $f$ on the maximum frequency $\max_{e \in E} |\mathcal{V}_e|$, and the upper bound $m$ on the maximum number of edges remain fixed. The edge-set $E$, on the other hand, keeps changing dynamically.   In the beginning, we have $E = \emptyset$. At each time-step, either an edge is inserted into the graph  or an edge is deleted from the graph. The goal is to maintain an approximately optimal solution to the problem in this dynamic setting.
\end{definition}

\section{Maintaining a Fractional Hypergraph $b$-Matching  in a Dynamic Setting}
\label{main:sec:set-cover}

\subsection{Preliminaries}

We first define a linear program for   fractional hypergraph $b$-matching (Definition~\ref{main:def:fractional:bmatching}). Next, we define the concept of a ``$\lambda$-maximal'' solution of this LP (Definition~\ref{main:def:maximal}) and prove the approximation guarantee for such a solution  (Theorem~\ref{main:th:maximal}). 
Our  main result  is summarized  in Theorem~\ref{main:th:main:result} and Corollary~\ref{main:cor:th:result}. 

Below, we  write  a linear program for a fractional hypergraph $b$-matching. 
\begin{eqnarray} 
\mbox{{\bf Primal LP:}} \qquad \mbox{Maximize } \qquad  \sum_{e \in E} x(e)  \label{main:lp:match-1} \label{lp:match-1} \\
\mbox{ subject to:} \qquad \sum_{e  \in   \mathcal{E}_v} x(e) \leq c_v \qquad  & \forall v \in V.  \label{main:eq:match-1} \label{eq:match-1} \\ 
0 \leq x(e) \leq \mu  \qquad  & \forall \text e \in E. 
\end{eqnarray}
\begin{eqnarray}
\mbox{{\bf Dual LP:}} \qquad \mbox{Minimize } \qquad  \sum_{v \in V} c_v \cdot y(v) + \sum_{e \in E} \mu \cdot z(e)  \label{main:dual:match-1} \label{dual:match-1} \\
\text{ subject to:} \qquad z(e) + \sum_{v \in \mathcal{V}_e} y(v) \geq 1\qquad  & \forall  e \in E.  \label{main:eq:dual:match-1}  \label{eq:dual:match-1} \\ 
y(v), z(e)  \geq 0 \qquad  & \forall  v \in V, e \in E. 
\end{eqnarray}

We next define the concept of a ``$\lambda$-maximal'' solution.

\begin{definition}
\label{main:def:maximal}
\label{def:maximal}
A feasible solution to LP~(\ref{main:lp:match-1}) is  $\lambda$-maximal (for $\lambda \geq 1$) iff for every edge $e \in E$ with $x(e) < \mu$, there is some node $v \in \mathcal{V}_e$ such that $\sum_{e' \in \mathcal{E}_{v}} x(e') \geq c_v/\lambda$. 
\end{definition}

\begin{theorem}
\label{main:th:maximal}
\label{th:maximal}
Let $f \geq \max_{e \in E} |\mathcal{V}_e|$ be an upper bound on the maximum possible ``frequency'' of an edge. Let OPT be the optimal objective value of LP~(\ref{main:lp:match-1}). Any $\lambda$-maximal solution to LP~(\ref{main:lp:match-1}) has an objective value that is at least  $\text{OPT}/(\lambda f+1)$. 
\end{theorem}

\begin{proof}
Let $\{ x^*(e) \}$ be a $\lambda$-maximal solution to the primal LP. Construct a dual solution $\{y^*(v), z^*(e)\},$ as follows. For every $v \in V$, set $y^*(v) = 1$ if $\sum_{e \in \mathcal{E}_v} x^*(e) \geq c_v/\lambda$, and $y^*(v) = 0$ otherwise. For every $e \in E$, set $z^*(e) = 1$ if $x^*(e) = \mu$ and $z^*(e) = 0$ otherwise. 

Consider the dual constraint corresponding to any edge $e' \in E$. Since the primal solution $\{x^*(e)\}$ is $\lambda$-maximal, either $x^*(e) = \mu$ or there is some $v' \in \mathcal{V}_{e'}$ for which $y^*(v') = 1$. In the former case we have $z^*(e) = 1$, whereas in the latter case we have $y^*(v') = 1$.  Hence, the dual constraint under consideration  is satisfied. This shows that the values $\{y^*(v),  z^*(e)\},$ constitute a feasible dual solution.  Next,  we infer that:
\begin{eqnarray}
& & \sum_{v \in V} c_v \cdot y^*(v) + \sum_{e \in E} \mu \cdot z^*(e) \nonumber \\
& = & \sum_{v \in V : y^*(v) = 1} c_v + \sum_{e \in E : z^*(e) = 1} \mu \label{eq:newlp:1} \\
& \leq & \sum_{v \in V : y^*(v) = 1} \lambda \cdot \sum_{e \in \mathcal{E}_v} x^*(e) + \sum_{e \in E : z^*(e) = 1} x^*(e) \label{eq:newlp:2} \\
& \leq & \sum_{v \in V} \lambda \cdot \sum_{e \in \mathcal{E}_v} x^*(e) + \sum_{e \in E} x^*(e) \nonumber \\
& \leq & \lambda \cdot f \cdot \sum_{e \in E} x^*(e) + \sum_{e \in E} x^*(e) \label{eq:newlp:3} \\
& = & (\lambda f + 1) \cdot \sum_{e \in E} x^*(e) \nonumber
\end{eqnarray}
Equation~\ref{eq:newlp:1} holds since $y^*(v) \in \{0,1\}$ for all $v \in V$ and $z^*(e) \in \{0,1\}$ for all $e \in E$. Equation~\ref{eq:newlp:2} holds since $y^*(v) = 1$ only if $\sum_{e \in \mathcal{E}_v} x^*(e) \geq c_v /\lambda$, and since $x^*(e) = \mu$ for all $e \in E$ with $z^*(e) = 1$. Equation~\ref{eq:newlp:3} holds since each edge can be incident upon at most $f$ nodes.

Thus, we have constructed a feasible dual solution whose objective is at most $(\lambda f+1)$-times the objective of the $\lambda$-maximal primal solution. The theorem now follows from weak duality. 
\end{proof}

Our main result is summarized below. For the rest of Section~\ref{main:sec:set-cover}, we focus on proving Theorem~\ref{main:th:main:result}.

\begin{theorem}
\label{main:th:main:result}
\label{th:main:result}
We can maintain a $(f+1 +\epsilon f)$-maximal  solution to the  dynamic fractional hypergraph $b$-matching problem  in $O(f \cdot \log (m+n)/\epsilon^2)$ amortized update time.   
\end{theorem}

\begin{corollary}
\label{main:cor:th:result}
We can maintain an $O(f^2)$-approximate solution to the dynamic hypergraph $b$-matching problem in $O(f \log (m+n)/\epsilon^2)$ amortized update time.
\end{corollary}

\begin{proof}
Follows from Theorem~\ref{main:th:maximal} and Theorem~\ref{main:th:main:result}.
\end{proof}

\subsection{The ($\alpha,\beta$)-partition and its properties.}
\label{sec:vc:partition}

For the time being, we restrict ourselves to the static setting.  Inspired by the primal-dual method for set-cover, we consider the following algorithm for the fractional hypergraph $b$-matching problem.
\begin{itemize}
\item Consider an initial  primal solution with $x(e) \leftarrow 0$ for all $e \in E$, and define $F \leftarrow E$.
\item {\sc While} there is some primal constraint that is not tight:
\begin{itemize}
\item Keep increasing the primal variables $\{x(e)\}, e \in F$, uniformly at the same rate till some primal constraint becomes tight. At that instant, ``freeze'' all the primal variables involved in that constraint and delete them from the set $F$, and set the corresponding dual variable to one. 
\end{itemize}
\end{itemize}

\noindent In Figure~\ref{main:fig:primaldual}, we  define a variant of the above procedure that happens to be easier to maintain in a dynamic setting. The main idea is to discretize the continuous primal growth process. Define $c_{\min} = \min_{v \in V} c_v$, and without  any loss of generality, assume that $c_{\min} > 0$. Fix two parameters $\alpha, \beta > 1$, and define $L = \lceil  \log_{\beta} (m \mu \alpha/c_{\min}) \rceil$. 

\begin{Claim}
\label{main:cl:discrete}
If we set $x(e) \leftarrow \mu \cdot \beta^{-L}$ for all $e \in E$, then we get a feasible primal solution.
\end{Claim}

\begin{proof}
Clearly, $x(e) \leq \mu$ for all $e \in E$. Now, consider any node $v \in V$. We have $\sum_{e \in \mathcal{E}_v} x(e) = |\mathcal{E}_v| \cdot \mu \cdot \beta^{-L} \leq |\mathcal{E}| \cdot \mu \cdot \beta^{-L} \leq m \cdot \mu \cdot \beta^{-L} \leq m \cdot \mu \cdot (c_{\min}/(m\mu \alpha)) = c_{\min}/\alpha < c_v$.  Hence, all the primal constraints are satisfied.
\end{proof}

\begin{figure}[htbp]
\centerline{\framebox{
\begin{minipage}{5.5in}
\begin{tabbing}
01. \ \ \ \  \=  Set $x(e) \leftarrow \mu \cdot \beta^{-L}$ for all $e \in E$, and define $c^*_v = c_v/(f \alpha \beta)$ for all $v \in V$.  \\
02.  \>  Set  $V_{L} \leftarrow \{ v \in V : \sum_{e \in \mathcal{E}_v} x(e) \geq c^*_v\}$, and $E_L \leftarrow \bigcup_{v \in V_L} \mathcal{E}_v$. \\
03. \> {\sc For} $i = L-1$ to $1$: \\
04. \> \ \ \ \ \ \ \= Set $x(e) \leftarrow x(e) \cdot \beta$ for all $e \in E \setminus \bigcup_{k=i+1}^L E_i$. \\
05. \> \> Set $V_{i} \leftarrow \left\{ v \in V \setminus \bigcup_{k=i+1}^L  V_k : \sum_{e \in \mathcal{E}_v} x(e) \geq c^*_v\right\}$. \\
06. \> \> Set $E_i \leftarrow  \bigcup_{v \in V_i} \mathcal{E}_v$. \\
07. \> Set $V_0 \leftarrow V \setminus \bigcup_{k=1}^L V_i$, and $E_0 \leftarrow  \bigcup_{v \in V_0} \mathcal{E}_v$. \\
08. \> Set $x(e) \leftarrow x(e) \cdot \beta$ for all $e \in E_0$. 
\end{tabbing}
\end{minipage}
}}
\caption{\label{main:fig:primaldual} DISCRETE-PRIMAL-DUAL().}
\end{figure}

Our new algorithm  is described in Figure~\ref{main:fig:primaldual}. We initialize our primal solution by setting $x(e) \leftarrow \mu \beta^{-L}$ for every edge $e \in E$, as per Claim~\ref{main:cl:discrete}. 
We call  a node $v$ {\em nearly-tight} if its corresponding primal constraint  is tight within a factor of $f\alpha\beta$, and {\em slack} otherwise. 
Furthermore, we call an edge {\em nearly-tight} if it is incident upon some nearly tight node, and {\em slack} otherwise. Let $V_L \subseteq V$ and $E_L \subseteq E$ respectively denote the  sets of nearly tight nodes and edges, immediately after the initialization step. The algorithm then performs  $L-1$ iterations.  

At iteration $i \in \{L-1, \ldots, 1\}$, the algorithm  increases the weight $x(e)$ of every slack edge $e$ by a factor of $\beta$. Since the total weight received by every slack node $v$ (from its incident edges) never exceeds $c_v/(f\alpha\beta)$, this weight-increase step does not violate any primal constraint. The algorithm then defines $V_i$ (resp. $E_i$) to be the set of new nodes (resp. edges) that become nearly-tight due to this weight-increase step. 

Finally, the algorithm defines $V_0$ (resp. $E_0$) to be the set of nodes (resp. edges) that are  slack at the end of iteration $i = 1$. It terminates after increasing the weight of every edge in $E_0$ by a factor of $\beta$.

When the algorithm terminates, it is easy to check that  $x(e) = \mu \cdot \beta^{-i}$ for every edge $e \in E_i$, $i \in \{0, \ldots, L\}$. 
We also have 
$c^*_v \leq \sum_{e \in \mathcal{E}_v} x(e) \leq \beta \cdot c^*_v$ for every node $v \in \bigcup_{k=1}^L V_k$, and $\sum_{e \in \mathcal{E}_v} x(e) \leq \beta \cdot c^*_v$ for every node $v \in V_0$. Furthermore, at the end of the algorithm, every edge $e \in E \setminus E_0$ is nearly-tight, and every edge $e \in E_0$ has weight $x(e) = \mu$.   We, therefore, reach the following conclusion.

\begin{Claim}
\label{main:cl:primaldual}
The algorithm described in Figure~\ref{main:fig:primaldual} returns an $(f\alpha\beta)$-maximal solution to the fractional hypergraph $b$-matching problem
with the additional property that $c^*_v \leq \sum_{e \in \mathcal{E}_v} x(e) \leq \beta \cdot c^*_v$ for every node $v \in \bigcup_{k=1}^L V_k$, and $\sum_{e \in \mathcal{E}_v} x(e) \leq \beta \cdot c^*_v$ for every node $v \in V_0$.

\end{Claim}

Our goal is to make a variant of the procedure in Figure~\ref{main:fig:primaldual} work in a dynamic setting. Towards this end, we  introduce the concept of an $(\alpha, \beta)$-partition (see Definition~\ref{def:vc:partition}) satisfying a certain invariant (see Invariant~\ref{inv:vc:1}). The reader is encouraged to notice the similarities between this construct and the output of the procedure in Figure~\ref{main:fig:primaldual}.

\begin{definition}
\label{def:vc:partition}
An {\em $(\alpha,\beta)$-partition} of the graph $G$ partitions its node-set $V$ into  subsets  $V_0 \ldots V_L$, where $L = \lceil\log_{\beta} (m \mu \alpha/c_{\min})\rceil$ and $\alpha, \beta > 1$. For $i \in \{0, \ldots, L\}$, we identify the subset $V_i$ as the $i^{th}$ ``level'' of this partition, and 
call $i$ the {\em level} $\ell(v)$ of a node $v$. 
We also define the level of each edge $e \in E$ as $\ell(e) = \max_{v \in \mathcal{V}_e} \left\{\ell(v) \right\}$, and assign a ``weight'' $w(e) = \mu \cdot \beta^{-\ell(e)}$ to the edge $e$. 
\end{definition}
  
 Given an $(\alpha,\beta)$-partition,
 let $\mathcal{E}_v(i) \subseteq \mathcal{E}_v$ denote the set of edges incident to $v$ that are in the $i^{th}$ level, and
 let $\mathcal{E}_v(i,j) \subseteq \mathcal{E}_v$ denote the set of edges incident to  $v$ whose levels are in the range $[i,j]$.
\begin{equation}
\mathcal{E}_v(i)  =   \{ e \in \mathcal{E}_v : \ell(e) = i\} \ \ \forall v \in V; i \in \{0,\ldots,L\} \label{eq:symbol:2} 
\end{equation}
\begin{equation}
\mathcal{E}_v(i,j)  =   \bigcup_{k = i}^j \mathcal{E}_v(k) \ \ \forall v \in V; i,j \in \{0,\ldots,L\}, i \leq j. \label{eq:symbol:3} 
\end{equation}
 
 Similarly,  we define the notations $\dd_v$ and $\dd_v(i,j)$. 
\begin{equation}
\dd_v   =   |\mathcal{E}_v|  \label{eq:symbol:4} 
\end{equation}
\begin{equation}
 \dd_v(i)   =   |\mathcal{E}_v(i)|  \label{eq:symbol:5} 
 \end{equation}
\begin{equation}
 \dd_v(i,j)  =  |\mathcal{E}_v(i,j)|  \label{eq:symbol:6}
 \end{equation}

Given an $(\alpha,\beta)$-partition, let $W_v = \sum_{e \in \mathcal{E}_v} w(e)$ denote the total weight a node $v \in V$ receives from the edges incident to it.  We also define the notation $W_v(i)$. It gives the total weight the node $v$ would receive from the edges incident to it, {\em if the node  $v$ itself were to go to the $i^{th}$ level}. Thus, we have $W_v = W_v(\ell(v))$. Since the weight of an edge $e$ in the hierarchical partition is given by $w(e) = \mu \cdot \beta^{-\ell(e)}$, we derive the following equations for all nodes $v \in V$.
\begin{equation}
W_v =  \sum_{e \in \mathcal{E}_v} \mu \cdot \beta^{-\ell(e)}. \label{eq:symbol:7} 
\end{equation}
\begin{equation}
W_v(i)  =   \sum_{e \in \mathcal{E}_v} \mu \cdot \beta^{-\max(\ell(e),i)} \ \ \forall i \in \{0,\ldots,L\}. \label{eq:symbol:8} 
\end{equation}

\begin{lemma}
\label{lm:partition}
An $(\alpha,\beta)$-partition satisfies the following conditions for all nodes $v \in V$.
\begin{equation}
W_v(L) \leq c_{\min}/\alpha \label{eq:lm:partition:1} 
\end{equation}
\begin{equation}
W_v(L)  \leq  \cdots  \leq W_v(i)  \leq \cdots \leq W_v(0)   \label{eq:lm:partition:2} 
\end{equation}
\begin{equation}
W_v(i) \leq \beta \cdot W_v(i+1)  \ \ \forall   i \in \{0,\ldots,L-1\}. \label{eq:lm:partition:3}
\end{equation}
\end{lemma}

\begin{proof}
Fix any $(\alpha,\beta)$-partition and any node $v \in V$.  We prove the first part of the lemma as follows.
\begin{eqnarray*}
W_v(L) = \sum_{e \in \mathcal{E}_v} \mu \cdot  \beta^{-\max(\ell(e),L)} 
 = \sum_{e \in \mathcal{E}_v} \mu \cdot \beta^{-L} \leq m \mu \cdot \beta^{-L} \leq m \mu \cdot \beta^{-\log_{\beta}(m \mu \alpha/c_{\min})} = c_{\min}/\alpha.
\end{eqnarray*}

We now fix any level $i \in \{0,\ldots, L-1\}$ and show that the $(\alpha,\beta)$-partition satisfies equation~\ref{eq:lm:partition:2}.
\begin{eqnarray*}
W_v(i+1) = \sum_{e \in \mathcal{E}_v} \mu \cdot  \beta^{-\max(\ell(e),i+1)} 
\leq \sum_{e \in \mathcal{E}_v} \mu \cdot \beta^{-\max(\ell(e),i)} = W_v(i).
\end{eqnarray*}

Finally, we prove equation~\ref{eq:lm:partition:3}.
\begin{eqnarray*}
W_v(i) = \sum_{e \in \mathcal{E}_v} \mu \cdot \beta^{-\max(\ell(e),i)} =\mu \cdot \beta \cdot  \sum_{e \in \mathcal{E}_v}  \beta^{-1-\max(\ell(e),i)} \\
\leq  \mu \cdot \beta \cdot  \sum_{e \in \mathcal{E}_v}  \beta^{-\max(\ell(e),i+1)} = \beta \cdot W_v(i+1)
\end{eqnarray*}
  \end{proof}

Fix any node $v \in V$, and focus on the value of $W_v(i)$ as we go down from the highest level $i = L$ to the lowest level $i = 0$.  Lemma~\ref{lm:partition} states that  $W_v(i) \leq c_{\min}/\alpha$ when $i = L$,  that $W_v(i)$ keeps increasing  as we go down the levels one after another,  and that $W_v(i)$ increases by at most a  factor of $\beta$ between consecutive levels. 

\medskip

We will maintain a specific type of $(\alpha,\beta)$-partition,  where each node is assigned to a level in  a way that satisfies 
the following Invariant~\ref{inv:vc:1}. This invariant is a relaxation of the bounds on
$\sum_{e \in  \mathcal{E}_v} x(e)$ for every node $v$ stated in Claim~\ref{main:cl:primaldual}.

\begin{invariant}
\label{inv:vc:1}
Define $c^*_v = c_v/(f \alpha \beta)$. For every node $v \in V \setminus V_0$, it holds that
$c^*_v \leq W_v \leq f \alpha \beta \cdot c^*_v$ and for every node $v \in V_0$ it holds that
$W_v \leq f \alpha \beta \cdot c^*_v$.
\end{invariant}

%The next theorem shows how to maintain a $2\alpha\beta$-approximate vertex cover in $G$ using Invariant~\ref{inv:vc:1}.

\begin{theorem}
\label{th:vc:structure}
Consider an $(\alpha,\beta)$-partition  that satisfies Invariant~\ref{inv:vc:1}. The edge-weights $\{w(e)\}, e\in E,$ give an $(f \alpha \beta)$-maximal solution to LP~(\ref{lp:match-1}).
\end{theorem}

\begin{proof}
 By Invariant~\ref{inv:vc:1}, we have $W_v \leq (f\alpha \beta) \cdot c^*_v = c_v$ for every node $v \in V$. Next, note that $w(e) \leq \mu$ for every edge $e \in E$. Thus, the weights $\{w(e)\}, e\in E,$ define a feasible solution to  LP~(\ref{lp:match-1}). 

 We claim that for every edge $e \in E$ with $w(e) < \mu$, there is some node $v \in \mathcal{V}_e$ for which $W_v \geq c_v/(f \alpha \beta)$.
 This will imply that the weights $\{w(e)\}, e \in E,$ form an  $(f \alpha\beta)$-maximal feasible solution to the primal LP.

To prove the claim, consider any edge $e \in E$ with $w(e) < \mu$. Since $w(e) = \mu \beta^{-\ell(e)}$, this implies that $\ell(e) > 0$. Let $v \in \arg \max_{u \in \mathcal{V}_e} \left\{ \ell(u) \right\}$.  Note that $\ell(e) = \ell(v)$. This implies that $\ell(v) > 0$. Hence, by Invariant~\ref{inv:vc:1}, we have $W_v  \geq c^*_v = c_v/(f \alpha \beta)$. This concludes the proof of the theorem.
\end{proof}

\subsection{The algorithm: Handling the insertion/deletion of an edge.}
\label{sec:vc:algo}

We now show how to maintain an $(\alpha, \beta)$-partition under edge insertions and deletions.
A node is called {\em dirty} if it violates Invariant~\ref{inv:vc:1}, and {\em clean} otherwise.
At the beginning of the algorithm the edge-set $E$ is  empty, and, thus,  every node is initially clean and at level zero. Now consider the time instant just prior to the $t^{th}$ update. By induction hypothesis, at this instant every node is clean. Then  the $t^{th}$ update takes place, which  inserts (resp. deletes)  an edge $e$ in $E$ with weight $w(e) = \mu \beta^{-\ell(e)}$. This increases (resp. decreases) the weights $\{W_v\}, v \in \mathcal{V}_e$. Due to this change, the nodes $v \in \mathcal{V}_e$ might become dirty.  To recover from this, we call the subroutine in Figure~\ref{fig:vc:dirty}, which works as follows

\begin{figure}[htbp]
\centerline{\framebox{
\begin{minipage}{5.5in}
\begin{tabbing}
01.   \=  {\sc While} there exists a dirty node  $v$ \\
02.  \>  \ \ \ \  \= {\sc If} $W_v > f \alpha \beta c^*_v$, {\sc Then} \\
\> \> \qquad  // {\em If true, then by equation~\ref{eq:lm:partition:1}, we have $\ell(v) < L$.} \\
03.  \> \> \ \ \ \ \ \ \ \ \ \= Increment the level of $v$   by setting $\ell(v) \leftarrow \ell(v)+1$. \\
04.  \> \> {\sc Else if} ($W_v < c^*_v$ and $\ell(v) > 0$), {\sc Then} \\
05.  \>  \> \> Decrement the level of $v$  by setting $\ell(v) \leftarrow \ell(v)-1$. 
\end{tabbing}
\end{minipage}
}}
\caption{\label{fig:vc:dirty} RECOVER().}
\end{figure}

Consider any node $v \in V$ and suppose that $W_v  > f \alpha \beta c^*_v = c_v \geq c_{\min}$. In this event, the algorithm increments the level of the node.
since $\alpha > 1$, equation~\ref{eq:lm:partition:1} implies that $W_v(L) < W_v(\ell(v))$ and, hence, we have $L > \ell(v)$. In other words, when the procedure described in Figure~\ref{fig:vc:dirty} decides to increment the level of a dirty node $v$ (Step 02), we know for sure that the current level of $v$ is strictly less than $L$ (the highest level in the $(\alpha,\beta)$-partition).

Next, consider an edge $e \in \mathcal{E}_v$. If we change $\ell(v)$, then this may change the weight $w(e)$,  and this in turn may change the weights $\{W_z\}, z \in \mathcal{V}_e$. Thus, a single iteration of the {\sc While} loop in Figure~\ref{fig:vc:dirty} may lead to some clean nodes becoming dirty, and some other dirty nodes becoming clean.   If and when  the {\sc While} loop terminates, however, we are guaranteed that every node is clean and that Invariant~\ref{inv:vc:1} holds.

\subsection{Data structures.}
\label{sec:vc:datastructures}

We now describe the relevant data structures that will be used by   our algorithm.

\begin{itemize}
\item We maintain for each node $v \in V$:
\begin{itemize}
\item A counter $\text{{\sc Level}}[v]$ to keep track of the current level of $v$. Thus, we set $\text{{\sc Level}}[v] \leftarrow \ell(v)$.
\item A counter $\text{{\sc Weight}}[v]$ to keep track of the weight of  $v$. Thus, we set $\text{{\sc Weight}}[v] \leftarrow W_v$.
\item For every level $i > \text{{\sc Level}}[v]$, we store  the set of edges $\mathcal{E}_v(i)$ in the form of a doubly linked list $\text{{\sc Incident-Edges}}_v[i]$.  For every level $i \leq \text{{\sc Level}}[v]$, the  list $\text{{\sc Incident-Edges}}_v[i]$ is empty.
\item For level $i = \text{{\sc Level}}[v]$, we store the set of edges $\mathcal{E}_v(0,i)$   in the form of a doubly linked list $\text{{\sc Incident-Edges}}_v[0,i]$.  For every level $i \neq \text{{\sc Level}}[v]$, the list $\text{{\sc Incident-Edges}}_v[0,i]$ is empty.
\end{itemize}
\item When the  graph gets updated due to an edge insertion/deletion, we may discover that a node violates Invariant~\ref{inv:vc:1}.  Recall that such a node is called {\em dirty}, and we store  the set of such nodes as a doubly linked list $\text{{\sc Dirty-nodes}}$. For every node $v \in V$, we maintain a bit $\text{{\sc Status}}[v] \in \{\text{dirty}, \text{clean}\}$ that indicates if the node is dirty or not. Every dirty node stores a pointer to its position in the list $\text{{\sc Dirty-nodes}}$.
\item The collection of linked lists $\bigcup_{i=0}^L \left\{ \text{\sc Incident-Edges}_v[0,i], \text{{\sc Incident-Edges}}_v[i]\right\}$ is denoted by the phrase {\em ``incidence lists of $v$''}.  For every edge $e \in E$, we maintain a counter $\text{{\sc Level}}[e]$ to keep track of $\ell(e)$. Furthermore, for every edge $e \in E$, we maintain $|\mathcal{V}_e|$ bidirectional pointers corresponding to the nodes in $\mathcal{V}_e$. The pointer corresponding to a node $v \in \mathcal{V}_e$ points  to the position of $e$ in the incidence lists of $v$.
Using these pointers, we  can update the incidence lists of the relevant nodes when the edge $e$ is inserted into (resp. deleted from) the graph, or when some node $v \in \mathcal{V}_e$ increases (resp. decreases) its level by one. 
\end{itemize}

\subsection{Bounding the amortized update time.} 
\label{sec:vc:updatetime}
We devote this section to the proof of  the following theorem. 

\begin{theorem}
\label{th:main:updatetime}
Fix any $\epsilon \in (0,1)$,  $\alpha = 1+1/f+3\eps$ and $\beta = 1+\eps$. Starting from an empty graph,  we can maintain an $(\alpha, \beta)$ partition in $G$ satisfying Invariant~\ref{inv:vc:1} in $O(f \log (m+n)/\eps^2)$ amortized update time. 
\end{theorem}

The main idea is as follows. After an edge insertion or deletion
 the data structure can be updated in time $O(1)$, plus the time to adjust the levels of the nodes, i.e., the time for procedure
RECOVER. To bound the latter quantity we note that each time the level of an edge $e \in E$ changes, we have to update at most $f$ lists (one corresponding to each node $v \in \mathcal{V}_e$). Hence, the time taken to update the lists is given by $f \cdot \delta_l$, where $\delta_l$ is the number of times the procedure in Figure~\ref{fig:vc:dirty} changes the level of an edge. Below, we show that $\delta_l \leq t \cdot O(L/\epsilon) = t \cdot O(\log (m+n)/\epsilon^2)$ after $t$ edge insertions/deletions in $G$ starting from an empty graph. This gives the required $O(f \delta_l/t) = O(f \log (m+n)/\epsilon^2)$ bound on the amortized update time.

Hence, to complete the proof of Theorem~\ref{th:main:updatetime}, we  need to give an amortized  bound on  {\em the number of times we have to change the level (or, equivalently, the weight) of an already existing edge}. During a single iteration of the WHILE loop in Figure~\ref{fig:vc:dirty}, this number is exactly $\dd_v(0,i)$ when node $v$ goes from level $i$ to level $i+1$, and at most $\dd_v(0,i)$ when node $v$ goes from level $i$ to  level $i-1$.

Specifically, we  devote the rest of this section to the proof of Theorem~\ref{th:runtime}, which implies that  on average we  change the weights of $O(L/\eps) = O(\log (m+n)/\eps^2)$ edges per update in $G$.

\begin{theorem}
\label{th:runtime}
Set $\alpha \leftarrow 1+1/f+3\eps$ and $\beta \leftarrow 1+\eps$. In the beginning, when  $G$ is an empty graph,   initialize  a counter $\text{\sc Count} \leftarrow 0$. Subsequently, each time we change the weight of an already existing edge in the hierarchical partition, set $\text{{\sc Count}} \leftarrow \text{{\sc Count}} + 1$.  Then $\text{{\sc Count}} = O(t  L/\eps)$ just  after we  handle  the $t^{th}$ update in $G$.
\end{theorem}

Recall that the level of an  edge $e$ is defined as $\ell(e) = \max_{v \in \mathcal{V}_e}(\ell(v))$.
%, and note that the weight  $w(y,z)$ decreases (resp. increases) iff the edge's level $\ell(y,z)$ goes up (resp. down). 
Consider the following thought experiment. We have a {\em bank account}, and initially, when there are no edges in the graph,   the bank account has a balance of zero dollars.  For each subsequent edge insertion/deletion, at most  $3L/\epsilon$ dollars are deposited to the bank account; and every time our algorithm changes the level of an already existing edge,   one dollar is withdrawn from it. We show that the bank account never runs out of money, and this implies that $\text{{\sc Count}} = O(t L/\epsilon)$ after $t$ edge insertions/deletions starting from an empty graph.

Let $\B$ denote the total amount of money (or potential) in the bank account at the present moment. We keep track of $\B$  by distributing an $\epsilon$-fraction of it among the nodes and the current set of edges in the graph. 
\begin{equation}
\label{eq:vc:potential:0}
\B = (1/\epsilon) \cdot \left(\sum_{e \in E} \Phi(e) + \sum_{v \in V} \Psi(v)\right)
\end{equation}

In the above equation, the amount of money (or potential) associated with an edge $e \in E$ is given by $\Phi(e)$, and the amount of money (or potential) associated with a node $v \in V$ is given by $\Psi(v)$.  At every point in time,  the potentials $\{\Phi(e), \Psi(v)\}$    will be determined by two invariants. But, before stating the invariants, we need to define the concepts of ``active'' and ``passive'' nodes.

\iffalse
We call a node $v \in V$  {\em passive} if we have $\mathcal{E}_v = \emptyset$ throughout the duration of a time interval that starts at the beginning of the algorithm (when $E = \emptyset$) and ends at  the present moment. Let $V_{passive} \subseteq V$  denote the set of all  nodes that are currently passive.
\fi

\iffalse
\begin{invariant}
\label{inv:vc:potential:node}
For every node $v \in V$, we have:
\begin{eqnarray*}
\Psi(v)  =  \begin{cases}
\epsilon \cdot (L- \ell(v))  + \left(\beta^{\ell(v)+1}/(f \mu (\beta-1))\right) \cdot \max\left(0,f \alpha c^*_v - W_v\right) & \text{ if } v \notin V_{\text{passive}}; \\
0 & \text{ if } v \in V_{\text{passive}}.
\end{cases}
\end{eqnarray*}
\end{invariant}
\fi

\begin{definition}
\label{def:kappa}
Consider any node $v \in V$. In the beginning, there is no edge incident upon the node $v$, and we initialize a counter $\kappa_v \leftarrow 0$. Subsequently, whenever an edge-insertion occurs in the graph, if the inserted edge is incident upon $v$, then we set $\kappa_v \leftarrow \kappa_v + 1$.  At any given time-step, we say that a node $v \in V$ is {\em active} if $\mu \kappa_v \geq c_v$ and 
{\em passive} otherwise.
\end{definition}

It is easy to check that if a node is active at time-step $t$, then it will remain active at every time-step $t' > t$. A further interesting consequence of the above definition is that a passive node is always at level zero, as shown in the lemma below.

\begin{lemma}
\label{lm:kappa}
At any given time-step, if a node $v \in V$ is passive, then we have $\ell(v) = 0$.
\end{lemma}

\begin{proof}
We prove this by induction. Let $\ell^{(t)}(v)$ and $\kappa_v^{(t)}$ respectively denote the level of the node $v$ and the value of the counter $\kappa_v$ at time-step $t$. Further, let $W_v^{(t)}$ denote the value of $W_v$ at time-step $t$. Initially, at time-step $t = 0$, the graph is empty, we have $W_v^{(0)} = 0$, and hence $\ell^{(0)}(v) = 0$. Now, by induction hypothesis, suppose that at time-step $t$ the node $v$ is passive and $\ell^{(t)}(v) = 0$, and, furthermore, suppose that the node $v$ remains passive at time-step $(t+1)$. Given this hypothesis, we claim that $\ell^{(t+1)}(v) = 0$. The lemma will follow if we can prove the claim.

To prove the claim, note that since the node $v$ is passive at time-step $(t+1)$, we have $\kappa_v^{(t+1)} \mu < c_v = f \alpha \beta c_v^*$. Since the node $v$ has at most $\kappa_v^{(t+1)}$ edges incident to it at time-step $(t+1)$, and since each of these edges has weight at most $\mu$, we have $W_v^{(t+1)} \leq \kappa_v^{(t+1)} \mu < f \alpha \beta c_v^*$. Now, recall Figure~\ref{fig:vc:dirty}. Since $\ell^{(t)}(v) = 0$ and since $W_v^{(t+1)} < f \alpha \beta c_v^*$, the node $v$ can never become dirty during the execution of the procedure in Figure~\ref{fig:vc:dirty} after the edge insertion/deletion that occurs at time-step $(t+1)$. Thus, the node $v$ will not change its level, and we will have $\ell^{(t+1)}(v) = 0$. This concludes the proof.   
\end{proof}

We are now ready to state the invariants that define edge and node potentials.

\begin{invariant}
\label{inv:vc:potential:edge}
For every edge $e \in E$, we have:
\begin{eqnarray*} \Phi(e)  = (1+\epsilon)  \cdot \left(L - \ell(e)\right)  \end{eqnarray*} 
\end{invariant}

\begin{invariant}
\label{inv:vc:potential:node}
Recall Definition~\ref{def:kappa}. For every node $v \in V$,  we have:
\begin{eqnarray*}
\Psi(v)  = \begin{cases} 
\left(\beta^{\ell(v)+1}/(f \mu (\beta-1))\right) \cdot \max\left(0,f \alpha \cdot c^*_v - W_v\right)  & \text{ if }  v \text{ is active}; \\
\left(\beta/(f  (\beta-1)\right) \cdot \kappa_v & \text{ otherwise.}
\end{cases}
\end{eqnarray*}
\end{invariant}

When the algorithm starts, the graph has zero edges, and all the nodes $v \in V$ are passive and at level $0$ with $W_v = 0$ and $\kappa_v = 0 < c_v/\mu$. At that moment, Invariant~\ref{inv:vc:potential:node} sets $\Psi(v) = 0$ for all nodes $v \in V$.  Consequently,  equation~\ref{eq:vc:potential:0} implies that  $\B = 0$. Theorem~\ref{th:runtime}, therefore, will follow if we can prove the next two lemmas. Their proofs appear in Section~\ref{subsec:vc:update} and Section~\ref{subsec:analyze:FIX} respectively. 

\begin{lemma}
\label{lm:update:special}
Consider the insertion (resp. deletion) of an edge $e$ in $E$. It  creates (resp. destroys) the weight $w(e) = \mu \cdot \beta^{-\ell(e)}$, creates (resp. destroys) the potential $\Phi(e)$, and changes the potentials $\{\Psi(v)\}, v \in \mathcal{V}_e$. Due to these changes,  the total potential $\B$ increases by at most $3L/\eps$.
\end{lemma}

\begin{lemma}
\label{lm:main:special}
During  every single iteration of the {\sc While} loop in Figure~\ref{fig:vc:dirty}, the total increase in {\sc Count} is no more than  the net decrease in the potential $\B$.
\end{lemma}

\subsection{Proof of Lemma~\ref{lm:update:special}.}
\label{subsec:vc:update}

\noindent {\bf Edge-insertion.}
Suppose that an edge $e$ is inserted into the graph at time-step $t$. Then the potential $\Phi(e)$ is created and gets a value of at most $(1+\eps)L$ units. Now, fix any node $v \in \mathcal{V}_e$, and consider three possible cases. 

\bigskip 
\noindent {\em Case 1.} The node $v$ was passive at time-step $(t-1)$ and remains passive at time-step $t$. In this case, due to the edge-insertion, the potential $\Psi(v)$ increases by $\beta/(f(\beta-1))$.

\bigskip
\noindent {\em Case 2.} The node $v$ was passive at time-step $(t-1)$ and becomes active at time-step $t$. In this case, we must have: $c_v - \mu \leq \mu \kappa_v^{(t-1)} < c_v \leq \mu \kappa_v^{(t)}$. By Invariant~\ref{inv:vc:potential:node}, just before the insertion of the edge $e$ we had:
 \begin{eqnarray}
 \Psi(v) & = & \left\{\beta/(f  \mu (\beta-1))\right\} \cdot \mu \kappa_v^{(t-1)} \nonumber \\
 & \geq & \left\{\beta/(f  \mu (\beta-1))\right\} \cdot (c_v - \mu)  \label{eq:veryverynew1}
 \end{eqnarray}
Since the node $v$ was passive at time-step $(t-1)$, by Lemma~\ref{lm:kappa} we infer that $\ell^{(t-1)}(v) = 0$. 
Hence, by Invariant~\ref{inv:vc:potential:node}, just after the insertion of the edge $e$ we get: 
\begin{eqnarray}
\Psi(v) & = &  \left\{\beta/(f \mu (\beta-1))\right\} \cdot \max\left(0,f \alpha \cdot c^*_v - W_v\right) \nonumber \\
& \leq &  \left\{\beta/(f \mu (\beta-1))\right\} \cdot (f \alpha c^*_v) \nonumber \\
& \leq & \left\{\beta/(f \mu (\beta-1))\right\} \cdot c_v \label{eq:veryverynew2}
\end{eqnarray} 
By equations~\ref{eq:veryverynew1},~\ref{eq:veryverynew2},  the potential $\Psi(v)$ increases by at most $\left\{\beta/(f \mu (\beta-1))\right\} \cdot (c_v - (c_v - \mu)) = \left\{\beta/(f(\beta-1))\right\}$.

\bigskip
\noindent {\em Case 3.} The node $v$ was active at time-step $(t-1)$. In this case, clearly the node $v$ remains active at time-step $t$, the weight $W_v$ increases, and hence the potential $\Psi(v)$ can only decrease.

\bigskip
\noindent From the above discussion, we conclude that the potential $\Psi(v)$ increases by at most  $\beta/(f(\beta-1))$ for every node $v \in \mathcal{V}_e$. Since $|\mathcal{V}_e| \leq f$, this accounts for a net increase of at most $f \cdot \beta/(f(\beta-1)) = \beta/(\beta-1) = \beta/\epsilon \leq L/\epsilon$. Finally, recall that the potential $\Phi(e)$ is created and gets a value of at most $(1+\epsilon) L \leq 2L/\epsilon$ units. Thus, the net increase in the potential $\B$ is at most $L/\epsilon + 2L/\epsilon = 3L/\epsilon$.

\bigskip
\noindent {\bf Edge-deletion.}
If an edge $e$ is deleted from $E$, then the potential $\Phi(e)$ is destroyed. The weight $W_v$ of each node $v \in \mathcal{V}_e$  decreases by at most $\mu \cdot \beta^{-\ell(v)}$. Furthermore, no passive node becomes active due to this edge-deletion, and, in particular, the counter $\kappa_v$ remains unchanged for every node $v \in V$. Hence,  each of the potentials $\{\Psi(v)\}, v \in \mathcal{V}_e,$ increases by at most $\beta^{\ell(v)+1}/(f \mu (\beta -1)) \cdot \mu \beta^{-\ell(v)} = \beta/(f (\beta-1))  = ((1+1/\epsilon)/f)  \le  2L/(\epsilon f)$. 
%The last inequality holds for $N \ge 8$  since $L = \log_{\beta} (n/\alpha) = \log_{1+\eps} (n/(1+2\eps)) $. 
 The potentials of the remaining nodes and edges do not change. Since $|\mathcal{V}_e| \leq f$, by equation~\ref{eq:vc:potential:0},  the net increase in   $\B$ is  at most $2 L/\epsilon \leq 3 L/\epsilon$.

\subsection{Proof of Lemma~\ref{lm:main:special}.}
\label{subsec:analyze:FIX}

Throughout this section, fix a single iteration of the {\sc While} loop in Figure~\ref{fig:vc:dirty} and suppose that it changes the level of a dirty node $v$ by one unit. We use the superscript $0$ (resp. $1$) on a symbol to denote its state at the time instant immediately prior to (resp. after) that specific iteration of the {\sc While} loop. Further,  we preface a symbol with  $\delta$ to denote the net decrease in its value due to that iteration. For example, consider the potential $\B$.  We have $\B = \B^0$ immediately before the iteration begins, and  $\B = \B^1$ immediately after  iteration ends.  We also have $\delta \B = \B^0 - \B^1$.

\iffalse
We will prove the following theorem. 
\begin{theorem}
\label{th:vc:analyze:FIX:main}
We have $\delta \B > 0$ and $\delta \B = \Omega(C)$, where $C$ denotes the runtime of FIX($v$). In other words, the money withdrawn from the bank account during the execution of FIX($v$) is sufficient to pay for  the computation performed by FIX($v$). 
\end{theorem}  
\fi

A change in the level of node $v$ does not affect the potentials of the edges  $e \in E \setminus \mathcal{E}_v$.  This observation, coupled with equation~\ref{eq:vc:potential:0}, gives us the following guarantee.

\begin{equation}
\label{eq:vc:change:1}
\delta \B  = (1/\epsilon) \cdot \left( \delta \Psi(v) +   \sum_{e \in \mathcal{E}_v} \delta \Phi(e) + \sum_{u \in V \setminus \{v\}}  \delta \Psi(u) \right)
\end{equation}

\medskip
\noindent {\bf Remark.} Since the node $v$ is changing its level, it must be active. Hence, by Invariant~\ref{inv:vc:potential:node}, we must have $\Psi(v) = \beta^{\ell(v) + 1}/(f \mu (\beta-1)) \cdot \max(0, f \alpha c^*_v - W_v)$.  We will use this observation multiple times throughout the rest of this section.

\bigskip
We divide the proof of Lemma~\ref{lm:main:special} into two possible cases, depending upon whether the concerned iteration of the {\sc While} loop increments or decrements the level of $v$. The main approach to the proof remains the same in each case. We first give an upper bound on the increase in $\text{{\sc Count}}$ due to the iteration. Next, we separately lower bound each of the following quantities:  $\delta \Psi(v)$,  $\delta \Phi(e)$ for all $e \in \mathcal{E}_v$, and   $\delta \Psi(u)$ for all $u \in V \setminus \{v\}$. Finally, applying equation~\ref{eq:vc:change:1}, we  derive that $\delta \B$ is sufficiently large to pay for the increase in $\text{{\sc Count}}$.

\bigskip
\noindent {\bf Remark.} Note that $\ell^0(u) = \ell^1(u)$ for all nodes $u \in V \setminus \{v\}$, and $\mathcal{E}^0_u = \mathcal{E}^1_u$ for all nodes $u \in V$. Thus, we will use the symbols $\ell(u)$ and $\mathcal{E}_u$ without any ambiguity for all such nodes.

\bigskip
\paragraph{Case 1: The  level of the node $v$ increases from $k$ to $(k+1)$.}
\label{sec:FIX:case2}

\begin{Claim}
\label{cl:verynew:1}
We have $\ell^0(e) = k$ and $\ell^1(e) = k+1$ for every edge $e \in \mathcal{E}_v^{0}(0,k)$.
\end{Claim}

\begin{proof}
Consider edge $e \in \mathcal{E}^0_v(0,k)$. Since $e \in \mathcal{E}^0_v(0,k)$, we have $\ell^0(e) \leq k$. Since $\ell^0(v) = k$ and $e \in \mathcal{E}_v$, we must have $\ell^0(e) = k$. Finally, since $\ell^1(u) = \ell^0(u)$ for all nodes $u \in V \setminus \{v\}$, we conclude that $\ell^1(e) = \ell^1(v) = k+1$. 
\end{proof}

\begin{Claim}
\label{cl:verynew:2}
We have $\ell^0(e) = \ell^1(e)$ for every edge $e \in \mathcal{E}_v^{0}(k+1, L)$.
\end{Claim}

\begin{proof}
Consider any edge $e \in \mathcal{E}_v^0(k+1, L)$. Since $\ell^0(e) \geq k+1$ and $\ell^0(v) = k$, there must be some node $u \in V \setminus \{v\}$ such that $\ell^0(u) \geq k+1$, $e \in \mathcal{E}_u$ and $\ell^0(e) = \ell^0(u)$. Since $\ell^1(u) = \ell^0(u) \geq k+1$ and  $\ell^1(v) = k+1$, we infer that $\ell^1(e) = \ell^1(u) = \ell^0(e)$. 
  \end{proof}

\begin{Claim}
\label{lm:FIX:case2:1}
We have $\text{{\sc Count}}^1 - \text{{\sc Count}}^0 =  \dd_v^0(0,k)$.
\end{Claim}

\begin{proof}
When the node $v$ changes its level from $k$ to $(k+1)$, this only affects the levels of those edges that are at level $k$ or below.
     \end{proof}

\begin{Claim}
\label{lm:FIX:case2:2}
We have $\delta \Psi(v) = 0$.
\end{Claim}

\begin{proof}
Since the node $v$ increases its level from $k$ to $(k+1)$, Step 02 (Figure~\ref{fig:vc:dirty}) guarantees that  $W_v^0 = W_v^0(k) > f \alpha \beta \cdot c^*_v$. Next, from Lemma~\ref{lm:partition} we infer that $W_v^1  = W_v^0(k+1) \geq \beta^{-1} \cdot W_v^0(k) > f \alpha c^*_v$. Since both $W_v^0, W_v^1 > f \alpha c^*_v$, we get:
$\Psi^0(v)  = \Psi^1(v) = 0$. It follows that $\delta \Psi(v) = \Psi^0(v) - \Psi^1(v) = 0$.
     \end{proof}

\begin{Claim}
\label{lm:FIX:case2:3}
For every edge $e \in \mathcal{E}_v$, we have:
\begin{eqnarray*}
\delta \Phi(e)  =
\begin{cases}
(1+\epsilon) &  \text{ if } e \in \mathcal{E}_v^0(0,k); \\
0 &  \text{ if } e \in \mathcal{E}_v^0(k+1,L). 
\end{cases} 
\end{eqnarray*} 
\end{Claim}

\begin{proof}
If $e \in \mathcal{E}_v^0(0,k)$, then  we have $\ell^0(e) = k$ and $\ell^1(e) = k+1$ (see Claim~\ref{cl:verynew:1}). Hence, we have $\Phi^0(e) = (1+\epsilon) \cdot (L - k)$ and $\Phi^1(e) =  (1+\epsilon) \cdot (L-k-1)$. It follows that $\delta \Phi(e) = \Phi^0(e) - \Phi^1(e) = (1+\epsilon)$. 

In contrast, if $e \in \mathcal{E}_v^0(k+1,L)$, then Claim~\ref{cl:verynew:2} implies that $\ell^0(e) = \ell^1(e) = l$ (say). Accordingly, we have $\Phi^0(e) = \Phi^1(e) = (1+\epsilon) \cdot (L - l)$. Hence, we get $\delta \Phi(e) = \Phi^0(e) - \Phi^1(e) = 0$.
     \end{proof}

\begin{Claim}
\label{lm:FIX:case2:4}
For every node $u \in V \setminus \{v\}$, we have:
\begin{eqnarray*}
\delta \Psi(u)  \geq
-(1/f) \cdot | \mathcal{E}_u \cap \mathcal{E}_v^0(0,k) | 
\end{eqnarray*} 
\end{Claim}

\begin{proof}
Consider any node $u \in V \setminus \{v\}$. If the node $u$ is passive, then we have $\delta \Psi(u) = 0$, and the claim is trivially true. Thus, for the rest of the proof we assume that the node $u$ is active.

Clearly, we have $\ell^0(e) = \ell^1(e)$ for each edge $e \in \mathcal{E}_u \setminus \mathcal{E}_v$. Hence, we get $\delta w(e) = 0$ for each edge $\mathcal{E}_u \setminus \mathcal{E}_v$. Next, by Claim~\ref{cl:verynew:2}, we have $\ell^0(e) = \ell^1(e)$ for each edge  $e \in \mathcal{E}_u \cap \mathcal{E}^0_v(k+1, L)$. Thus, we get $\delta w(e) = 0$ for each edge $e \in \mathcal{E}_u \cap \mathcal{E}^0_v(k+1,L)$. We therefore conclude  that: 
\begin{eqnarray}
\delta W_u & = & \sum_{e \in \mathcal{E}_u \setminus \mathcal{E}_v} \delta w(e) + \sum_{e \in \mathcal{E}_u \cap \mathcal{E}^0_v(k+1,L)} \delta w(e) + \sum_{e \in \mathcal{E}_u \cap \mathcal{E}_v^0(0,k)} \delta w(e) \nonumber \\
& = & \sum_{e \in \mathcal{E}_u \cap \mathcal{E}_v^0(0,k)} \delta w(e) \nonumber \\
& = & |\mathcal{E}_u \cap \mathcal{E}_v^0(0,k)| \cdot \mu \cdot (\beta^{-k} - \beta^{-(k+1)}) \nonumber \\
& = & |\mathcal{E}_u \cap \mathcal{E}_v^0(0,k)| \cdot \mu \cdot (\beta-1)/\beta^{k+1} \nonumber
\end{eqnarray} Using this observation, we infer that:  
\begin{eqnarray}
\delta \Psi(u) & \geq & - \left(\beta^{\ell(u)+1}/(f \mu (\beta-1))\right) \cdot \delta W_u \nonumber \\
& = & - \left(\beta^{\ell(u)+1}/(f \mu (\beta-1))\right) \cdot  |\mathcal{E}_u \cap \mathcal{E}_v^0(0,k)| \cdot \mu \cdot (\beta-1)/\beta^{k+1} \nonumber \\
& \geq & -\beta^{\ell(u)-k} \cdot (1/f) \cdot |\mathcal{E}_u \cap \mathcal{E}_v^0(0,k)|  \nonumber \\
& \geq & - (1/f) \cdot |\mathcal{E}_u \cap \mathcal{E}_v^0(0,k)| \label{eq:verynew11}
 \end{eqnarray}
Equation~\ref{eq:verynew11} holds since either $|\mathcal{E}_u \cap \mathcal{E}_v^0(0,k)| = 0$, or there is an edge $e \in \mathcal{E}_u \cap \mathcal{E}_v^0(0,k)$. In the former case, equation~\ref{eq:verynew11} is trivially true. In the latter case, by Claim~\ref{cl:verynew:1} we have $\ell^0(e) = k$, and since $\ell^0(e) \geq \ell(u)$, we infer that $\ell(u) \leq k$ and  $\beta^{\ell(u) - k} \leq 1$. 
     \end{proof}
  
\begin{Claim}
\label{lm:FIX:case2:5}
We have:
\begin{eqnarray*}
\sum_{u \in V \setminus \{v\}} \delta \Psi(u)  \geq - \dd_v^0(0,k)
\end{eqnarray*} 
\end{Claim}

\begin{proof}
We have:
\begin{eqnarray}
\sum_{u \in V \setminus \{v\}} \delta \Psi(u)  & = & \sum_{u \in V \setminus \{v\} : \mathcal{E}_u \cap \mathcal{E}_v^0(0,k) \neq \emptyset} \delta \Psi(u) \label{eq:1}\\
& \geq & \sum_{u \in V \setminus \{v\} : \mathcal{E}_u \cap \mathcal{E}_v^0(0,k) \neq \emptyset} - (1/f) \cdot |\mathcal{E}_u \cap \mathcal{E}_v^0(0,k)| \label{eq:2}\\
& \geq & \sum_{e \in \mathcal{E}_v^0(0,k)} f \cdot (-1/f) \label{eq:3} \\
& = & - \dd_v^0(0,k) \nonumber 
\end{eqnarray}
Equations~\ref{eq:1} and~\ref{eq:2} follow from Claim~\ref{lm:FIX:case2:4}. Equation~\ref{eq:3} follows from a simple counting argument and the fact that the maximum frequency of an edge is $f$. 
\end{proof}

\noindent From Claims~\ref{lm:FIX:case2:2},~\ref{lm:FIX:case2:3},~\ref{lm:FIX:case2:5} and equation~\ref{eq:vc:change:1}, we derive the following bound.
\begin{eqnarray}
\delta \B & = & (1/\epsilon) \cdot \left(\delta \Psi(v) + \sum_{e \in \mathcal{E}_v} \delta \Phi(e) + \sum_{u \in V \setminus \{v\}} \delta \Psi(u)\right)  \nonumber \\
& \geq & (1/\epsilon) \cdot \left(0 + (1+\epsilon) \cdot D_v^0(0,k) -  D_v^0(0,k) \right) \nonumber \\
&  = &    D_v^0(0,k) \nonumber
\end{eqnarray}

Thus, Claim~\ref{lm:FIX:case2:1} implies that the net decrease in the potential $\B$ in no less than the increase in $\text{{\sc Count}}$. This proves Lemma~\ref{lm:main:special} for Case 1.

\bigskip
\paragraph{Case 2:  The  level of the node $v$ decreases from $k$ to  $k-1$.}
\label{sec:FIX:case3}

\begin{Claim}
\label{cl:verynew:3}
For every edge $e \in \mathcal{E}_v^0(0,k)$, we have $\ell^0(e) = k$ and $w^0(e) = \mu \beta^{-k}$.
\end{Claim}

\begin{proof}
Consider any edge $e \in \mathcal{E}_v^0(0,k)$. Using the same argument as in the proof of Claim~\ref{cl:verynew:1}, we can show that $\ell^0(e) = k$. Since $\ell^0(e) = k$, we must have $w^0(e) = \mu \beta^{-k}$.
  \end{proof}
The next claim bounds the degree $D_v^0(0,k)$ of node $v$, which we then use in the following claim to bound the increase in 
$ \text{{\sc Count}}$.

\begin{Claim}
\label{cl:FIX:case3:degree}
We have $W_v^0 = W_v^0(k) < c^*_v$, and, furthermore, $D_v^0(0,k) \leq  \beta^k c^*_v/\mu$.
\end{Claim}

\begin{proof}
Since the node $v$ decreases its level from $k$ to $(k-1)$, Step~04 (Figure~\ref{fig:vc:dirty}) ensures that $W_v^0 = W_v^0(k) < c^*_v$. Claim~\ref{cl:verynew:3} implies that $w^0(e) = \mu \beta^{-k}$ for all $e \in \mathcal{E}_v^0(0,k)$. We conclude that: 
$$c^*_v > W_v^0 \geq \sum_{e \in \mathcal{E}_v^0(0,k)} w^0(e) = \mu \beta^{-k} \cdot D_v^0(0,k).$$ Thus, we get $D_v^0(0,k) \leq c^*_v \beta^k/\mu$.
     \end{proof}

\begin{Claim}
\label{lm:case3:runtime}
We have $\text{{\sc Count}}^1 - \text{{\sc Count}}^0 \leq  c^*_v \beta^k/\mu$.
\end{Claim}

\begin{proof}
The node $v$ decreases its level from $k$ to $k-1$. Due to this event, the level of an edge changes only if it belongs to $\mathcal{E}_v^0(0,k)$. Thus, we have $\text{{\sc Count}}^1 - \text{{\sc Count}}^0  \leq D_v^0(0,k) \leq c^*_v \beta^k/\mu$.
     \end{proof}

\begin{Claim}
\label{lm:FIX:case3:node:u}
For all $u \in V \setminus \{v\}$, we have $\delta \Psi(u)  \geq 0$. 
\end{Claim}

\begin{proof}
Fix any node $u \in V \setminus \{v\}$. If the node $u$ is passive, then we have $\delta \Psi(u) = 0$, and the claim is trivially true. Thus, for the rest of the proof we assume that the node $u$ is active.

 If $\mathcal{E}_u \cap \mathcal{E}_v^0(0,k) = \emptyset$, then we have $W^0_u = W^1_u$, and hence, $\delta \Psi(u) = 0$. Else we have $\mathcal{E}_u \cap \mathcal{E}_v^0(0,k) \neq \emptyset$. In this case, as the level of the node $v$ decreases from $k$ to $k-1$, we infer that $w^0(e) \leq w^1(e)$ for all $e \in \mathcal{E}_u \cap \mathcal{E}_v^0(0,k)$, and, accordingly, we get $W_u^0 \leq W_u^1$. This implies that $\Psi^0(u) \geq \Psi^1(u)$. Thus, we have $\delta \Psi(u) = \Psi^0(u) - \Psi^1(u) \geq 0$.
     \end{proof}

We now partition the edge-set $\mathcal{E}_v$ into two subsets, $X$ and $Y$, according to the level of the other endpoint.
$$X = \left\{ e \in \mathcal{E}_v : \max_{u \in \mathcal{V}_e \setminus \{v\}} \left\{ \ell(u) \right\} < k\right\} \text{ and } Y = \mathcal{E}_v \setminus X.$$ 

\begin{Claim}
\label{lm:FIX:case3:edge}
For every edge $e \in \mathcal{E}_v$, we have:
\begin{eqnarray*}
\delta \Phi(e)  = 
\begin{cases}
0 &  \text{ if } e \in Y; \\
-(1+\epsilon)  &  \text{ if } e \in X.
\end{cases} 
\end{eqnarray*} 
\end{Claim}

\begin{proof}
Fix any edge $e \in \mathcal{E}_v$. We consider two possible scenarios.
\begin{enumerate}
\item We have $e \in Y$. As  the level of the node $v$ decreases from $k$ to $k-1$,  we infer that $\ell^0(e) = \ell^1(e)$, and accordingly, $\Phi^0(e) = \Phi^1(e)$. Hence, we get $\delta \Phi(e) = \Phi^0(e) - \Phi^1(e) = 0$.
\item We have $e \in X$. Since  the level of node $v$ decreases from $k$ to $k-1$,  we infer that $\ell^0(e) = k$ and $\ell^1(e) = k-1$, and accordingly, $\Phi^0(e) =  (1+\epsilon) \cdot (L - k)$ and $\Phi^1(e) = (1+\epsilon) \cdot (L-k+1)$. Hence, we get $\delta \Phi(e) = \Phi^0(e) - \Phi^1(e) = -(1+\epsilon)$.
\end{enumerate}
This concludes the proof of the Claim.
     \end{proof}

Next, we partition $W_v^0$ into two parts: $x$ and $y$. The first part denotes the contributions towards $W_v^0$ by the edges $e \in X$, while the second part denotes the contribution towards $W_v^0$ by the edges $e \in Y$. Note that $X \subseteq \mathcal{E}_v^0(0,k)$, which implies that $x = \sum_{e \in X} w^0(e) = \mu \beta^{-k} \cdot |X|$. Thus, we get the following equations.
\begin{eqnarray}
W_v^0  =  x+ y < c^*_v \label{eq:part:1} \\
x   = \mu   \beta^{-k} \cdot |X| \label{eq:part:2} \\
y  =  \sum_{e \in Y} w^0(e) \label{eq:part:3}
\end{eqnarray}

Equation~\ref{eq:part:1} holds due to Claim~\ref{cl:FIX:case3:degree}.

\begin{Claim}
\label{lm:FIX:case3:edge:sum}
We have $\sum_{e \in \mathcal{E}_v} \delta \Phi(e) =  -(1+\epsilon) \cdot x \cdot \beta^k/\mu$.
\end{Claim}

\begin{proof}
Claim~\ref{lm:FIX:case3:edge} implies that $\sum_{e \in \mathcal{E}_v} \delta \Phi(e) = -(1+\epsilon) \cdot |X|$. Applying equation~\ref{eq:part:2}, we infer that $|X| =  x \cdot \beta^k/\mu$. 
     \end{proof}

\begin{Claim}
\label{lm:new:1}
We have: 
\begin{eqnarray*}
\delta \Psi(v) =  (f \alpha c^*_v - x-y) \cdot  \frac{\beta^{k+1}}{f \mu (\beta-1)} 
- \max\left(0,f \alpha c^*_v  - \beta x - y\right) \cdot \frac{\beta^{k}}{f \mu (\beta-1)}.
\end{eqnarray*}
\end{Claim}

\begin{proof}
Equation~\ref{eq:part:1} states that $W_v^0 = x+y < c^*_v$. Since $\ell^0(v) = k$, we get:
\begin{equation}
\label{eq:FIX:case3:deltav:1}
\Psi^0(v) =  (f \alpha c^*_v - x - y) \cdot \frac{\beta^{k+1}}{f \mu (\beta-1)}  
\end{equation}
As the node $v$ decreases its level from $k$ to $k-1$, we have: 
\begin{eqnarray*}
w^1(e)  = \begin{cases} 
\beta \cdot w^0(e) & \text{ if } e \in X; \\ 
 w^0(e) & \text{ if } u \in Y
\end{cases}
\end{eqnarray*} 
Accordingly, we have $W_v^1 = \beta \cdot x + y$, which implies the following equation.
\begin{equation}
\label{eq:FIX:case3:deltav:2}
\Psi^1(v) =  \max(0,f \alpha c^*_v - \beta x -y) \cdot \frac{\beta^{k}}{f\mu(\beta-1)} 
\end{equation}
Since $\delta \Psi(v) = \Psi^0(v) - \Psi^1(v)$, the Claim  follows from equations~\ref{eq:FIX:case3:deltav:1} and~\ref{eq:FIX:case3:deltav:2}.
     \end{proof}

\noindent We now consider two possible scenarios depending upon the value of $(f \alpha c^*_v - \beta x - y)$. We show that in each case  $\delta \B \geq c^*_v \beta^k/\mu$. This, along with Claim~\ref{lm:case3:runtime}, implies that $\delta \B \geq \text{{\sc Count}}^1 - \text{{\sc Count}}^0$. This proves Lemma~\ref{lm:main:special} for Case 2.
\begin{enumerate}
\item Suppose that $(f \alpha c^*_v - \beta x - y) < 0$. From  Claims~\ref{lm:FIX:case3:node:u},~\ref{lm:FIX:case3:edge:sum},~\ref{lm:new:1} and equation~\ref{eq:vc:change:1}, we derive:
\begin{eqnarray}
\epsilon \cdot \delta \B & = & \sum_{u \in V \setminus \{v\}} \delta \Psi(u) + \sum_{e \in \mathcal{E}_v} \delta \Phi(e) + \Psi(v) \nonumber \\
& \geq &  - (1+\epsilon) \cdot x \cdot \frac{\beta^k}{\mu}  + (f \alpha c^*_v - x - y) \cdot \frac{\beta^{k+1}}{f\mu(\beta-1)} \nonumber \\
& \geq &  - (1+\epsilon)  \cdot c^*_v \cdot \frac{\beta^k}{\mu} + (f \alpha -1)  c^*_v \cdot \frac{\beta^{k+1}}{f\mu(\beta-1)} \label{eq:very:1} \\
& = & \frac{c^*_v \beta^k}{\mu} \left\{ - (1+\epsilon) + (\alpha - 1/f) \cdot \frac{\beta}{(\beta-1)} \right\} \nonumber \\
& = & \frac{c^*_v \beta^k}{\mu} \left\{ - (1+\epsilon) + (1+3\epsilon) \cdot \frac{(1+\epsilon)}{\epsilon} \right\} \label{eq:very:2} \\
& \geq & \epsilon \cdot c^*_v \cdot \frac{\beta^k}{\mu} \nonumber 
\end{eqnarray}
Equation~\ref{eq:very:1} follows from equation~\ref{eq:part:1}. Equation~\ref{eq:very:2} holds since $\alpha = 1+1/f + 3\eps$ and $\beta = 1+\eps$.  \\

\item Suppose that $(f \alpha c^*_v - \beta x - y) \geq 0$. From  Claims~\ref{lm:FIX:case3:node:u},~\ref{lm:FIX:case3:edge:sum},~\ref{lm:new:1} and equation~\ref{eq:vc:change:1}, we derive:
\begin{eqnarray}
& & \epsilon \cdot \delta \B  =  \sum_{u \in V \setminus \{v\}} \delta \Psi(u) + \sum_{e \in \mathcal{E}_v} \delta \Phi(u,v) + \Psi(v) \nonumber \\
 & & \geq  - (1+\epsilon) \cdot x \cdot \frac{\beta^k}{\mu}  + (f \alpha c^*_v - x - y) \cdot \frac{\beta^{k+1}}{f\mu(\beta-1)}  -  (f \alpha c^*_v - \beta x - y) \cdot \frac{\beta^{k}}{f\mu(\beta-1)} \nonumber \\
& & =    \frac{\beta^k}{\mu(\beta-1)} \cdot   \big\{(f\alpha c^*_v - x - y) \cdot \frac{\beta}{f}    - (f\alpha c^*_v -\beta x -y) \cdot \frac{1}{f}  - (1+\epsilon) \cdot x \cdot (\beta-1)  \big\} \nonumber \\
&  & = \frac{\beta^k}{\mu(\beta-1)} \cdot \big\{\alpha c^*_v \beta -\alpha c^*_v - \frac{(\beta x + \beta y - \beta x - y)}{f}     - (1+\epsilon) \cdot x \cdot (\beta-1)  \big\} \nonumber \\
& & = \frac{\beta^k}{\mu(\beta-1)} \cdot \big\{\alpha c^*_v  \cdot (\beta -1) - \frac{y( \beta   - 1)}{f}     - (1+\epsilon) \cdot x \cdot (\beta-1)  \big\}  \nonumber \\
& & =     \frac{\beta^k}{\mu} \cdot \big\{\alpha c^*_v    - \frac{y}{f} -  (1 + \eps) \cdot x\big\} \nonumber \\ 
& & \geq     \frac{\beta^k}{\mu} \cdot \big\{\alpha c^*_v    - \beta (y +  x)\big\} \label{eq:verynew:1} \\ 
& & \geq     \frac{\beta^k}{\mu} \cdot (\alpha    - \beta) \cdot c^*_v \label{eq:verynew:2} \\ 
& & \geq  \epsilon \cdot c^*_v \cdot \frac{\beta^k}{\mu}  \label{eq:verynew:3}
\end{eqnarray}
Equation~\ref{eq:verynew:1} holds since $\beta = 1 + \eps$ and $ f \geq 1$. Equation~\ref{eq:verynew:2} follows from Equation~\ref{eq:part:1}. Equation~\ref{eq:verynew:3} holds since $\alpha = 1 + 1/f + 3\epsilon$ and $\beta = 1+\epsilon$. 
\end{enumerate}

\section{Maintaining a Set-Cover in a Dynamic Setting}
\label{sec:set-cover}

We  first show the  link  between the fractional hypergraph $b$-matching  and   set-cover.

\begin{lemma}
\label{main:lm:set-cover}
The dual LP~(\ref{main:dual:match-1}) is an  LP-relaxation of the set-cover problem (Definition~\ref{main:def:set-cover}). 
\end{lemma}

\begin{proof}
Given an instance of the set-cover problem, we create an instance of the hypergraph $b$-matching problem as follows. For each element $u \in \mathcal{U}$ create an edge $e(u) \in E$, and for each set $S \in \mathcal{S}$, create a node $v(S) \in V$ with cost $c_{v(S)} = c_S$. Ensure that an element $u$ belongs to a set $S$ iff $e(u) \in \mathcal{E}_{v(S)}$.  Finally, set $\mu = \max_{v \in V} c_v +1$.

Since $\mu > \max_{v \in V} c_v$, it can be shown that  an optimal solution to the dual LP~(\ref{main:dual:match-1}) will set $z(e) = 0$ for every edge $e \in E$. Thus, we can remove the variables $\{z(e)\}$ from the constraints and the objective function of LP~(\ref{main:dual:match-1}) to get a new LP with the same optimal objective value. This new LP is an LP-relaxation for the set-cover problem. 
\end{proof}

We now present the main result of this section.

\begin{theorem}
\label{main:cor:set-cover}
We can maintain an $(f^2 + f + \epsilon f^2)$-approximately optimal solution to the dynamic set cover problem   in $O(f \cdot \log (m+n)/\epsilon^2)$ amortized update time.
\end{theorem}

\begin{proof}
We  map the set cover instance to a fractional hypergraph $b$-matching instance as in the proof of Lemma~\ref{main:lm:set-cover}. By Theorem~\ref{main:th:main:result}, in $O(f \log (m+n)/\epsilon^2)$ amortized update time, we can maintain a feasible solution $\{x^*(e)\}$ to LP~(\ref{main:lp:match-1}) that is $\lambda$-maximal, where $\lambda = f+1+\epsilon f$.  

Consider a collection of sets $\mathcal{S}^* = \{ S \in \mathcal{S} : \sum_{e \in \mathcal{E}_{v(S)}} x^*(e) \geq c_{v(S)}/\lambda\}$. Since we can maintain the fractional solution $\{x^*(e)\}$ in $O(f \log (m+n)/\epsilon^2)$ amortized update time, we can also maintain  $\mathcal{S}^*$ without incurring any additional overhead in the update time. Now, using complementary slackness conditions, we can show that  each element  $e \in \mathcal{U}$ is covered by some $S \in \mathcal{S}^*$, and  the sum $\sum_{S \in \mathcal{S}^*} c_S$ is at most $(\lambda f)$-times the size of the primal solution $\{x^*(e)\}$. The corollary follows from LP duality.
\end{proof}

\section{Maintaining a $b$-Matching in a Dynamic Setting}
\label{sec:bmatching}

We will present a dynamic algorithm for  maintaining an $O(1)$-approximation to the maximum $b$-matching (see Definitions~\ref{main:def:b-matching},~\ref{main:def:dynamic:b-matching}). Our main result is summarized in  Theorem~\ref{th:sample:main}.  We use the following approach.  First,  we note that the fractional $b$-matching problem is a special case of the fractional hypergraph $b$-matching problem (see Definition~\ref{main:def:fractional:bmatching}) with $f = 2$ (for each edge is incident upon exactly two nodes). Hence, by Theorems~\ref{main:th:maximal} and~\ref{main:th:main:result}, we can maintain a $O(f^2) = O(1)$ approximate ``fractional'' solution to the maximum $b$-matching problem in $O(f \log (m+n)) = O(\log n)$ amortized update time. Next, we perform  randomized rounding on this fractional solution in the dynamic setting, whereby we select each edge in the solution with some probability that is determined by its fractional value. This leads to Theorem~\ref{th:sample:main}.

\paragraph{Notations.} Let $G = (V, E)$ be the input graph to the $b$-matching problem. Given any subset of edges $E' \subseteq E$ and any node $v \in V$, let $\N(v, E') = \{ u \in V : (u,v) \in E'\}$ denote the set of neighbors of $v$ with respect to the edge-set $E'$, and let $\text{deg}(v, E') = |\N(v, E')|$. Next, consider any ``weight'' function $w : E' \rightarrow \mathbf{R}^+$ that assigns a weight $w(e)$ to every edge $e \in E'$. For every node $v \in V$, we define $W_v = \sum_{u \in \N(v, E)} w(u,v)$. Finally, for every subset of edges $E' \subseteq E$, we define $w(E') = \sum_{e \in E'} w(e)$.

\medskip
 Recall that in the $b$-matching problem, we are given an ``input graph'' $G = (V, E)$ with $|V| = n$ nodes, where each node $v \in V$ has a ``capacity'' $c_v \in \{1, \ldots, n\}$. We want to select a subset $E' \subseteq E$ of edges of maximum size such that each node $v$ has at most $c_v$ edges incident to it in $E'$. We will also be interested in ``fractional'' $b$-matchings. In the fractional $b$-matching problem, we want to assign a weight $w(e) \in [0,1]$ to every edge $e \in E$ such that $\sum_{u \in \N(v, E)} w(u,v) \leq c_v$ for every node $v \in V$, and the sum of the edge-weights $w(E)$ is maximized. In the dynamic version of these problems, the node-set $V$ remains fixed, and at each time-step the edge-set $E$ gets updated due to an edge insertion or deletion. We now show how to efficiently maintain an $O(1)$-approximate fractional $b$-matching in the dynamic setting.

\begin{theorem}
\label{th:sample:b-matching}
Fix a constant  $\epsilon \in (0,1/4)$, and let $\lambda = 4$, and $\gamma = 1+4\epsilon$. In $O(\log n)$ amortized update time, we can maintain a fractional $b$-matching $w :  E \rightarrow [0,1]$ in  $G = (V,E)$ such that:
\begin{eqnarray}
\label{main:eq:w:1}
W_v \leq c_v/\gamma \text{ for all nodes } v \in V.  \label{eq:w:1} \\
w(u,v) = 1 \text{ for each edge } (u,v) \in E \text{ with } W_u, W_v < c_v/\lambda. \label{eq:w:2}
\end{eqnarray} 
Further, the size of the optimal $b$-matching in $G$ is $O(1)$ times the sum $\sum_{e \in E} w(e)$. 
\end{theorem}

\begin{proof}
Note that the fractional $b$-matching problem is a special case of  fractional hypergraph $b$-matching where $\mu = 1$, $m = n^2$, and $f = 2$. 

We scale down the capacity of each node $v \in V$ by a factor of $\gamma$, by defining $\tilde{c}_v = c_v/\gamma$ for all $v \in V$.  Next, we apply Theorem~\ref{th:main:result} on the input simple graph $G = (V, E)$ with $\mu = 1$, $m = n^2$, $f = 2$, and  the reduced capacities $\{\tilde{c}_v\}, v \in V$. Let $\{w(e)\}, e\in E,$ be the resulting $(f+1+\epsilon f)$-maximal matching (see Definition~\ref{def:maximal}). Since  $\epsilon < 1/3$ and $f = 2$, we have $\lambda \geq f+1+\epsilon f$. Since $\epsilon$ is a constant, the amortized update time for maintaining the fractional $b$-matching becomes $O(f \cdot \log (m+n)/\epsilon^2) = O(\log n)$. Finally, by Theorem~\ref{th:maximal}, the fractional $b$-matching $\{w(e)\}$ is an $(\lambda f +1) = 9$-approximate optimal $b$-matching in $G$ in the presence of the reduced capacities $\{\tilde{c}_v\}$. But scaling down the capacities reduces the objective of LP~(\ref{lp:match-1}) by at most a factor of $\gamma$. Hence, the size of the optimal $b$-matching in $G$ is at most $9\gamma = O(1)$ times the sum $\sum_{e \in E} w(e)$. This concludes the proof.
 \ \end{proof}

Set $\lambda = 4$, $\gamma = 1+4\epsilon$ and $\epsilon \in (0,1/4)$ for the rest of this section. We will show how to dynamically convert  the fractional $b$-matching $\{w(e)\}$ from Theorem~\ref{th:sample:b-matching} into an integral $b$-matching, by losing a constant factor in the approximation ratio. The main idea is to randomly sample the edges $e \in E$ based on their $w(e)$ values. But, first we  introduce the following notations.

Say that a node $v \in V$ is ``nearly-tight'' if $W_v \geq c_v/\lambda$ and ``slack'' otherwise. Let $T$ denote the set of all nearly-tight  nodes.  We also  partition of the node-set $V$ into two subsets: $B \subseteq V$ and $S = V \setminus B$. Each node $v \in B$ is called ``big'' and has $\text{deg}(v, E) \geq c \log n$, for some large constant $c > 1$. Each node $v \in S$ is called ``small'' and has $\text{deg}(v, E) < c \log n$. Define $E_B = \{ (u,v) \in E : \text{either } u \in B \text{ or } v \in  B\}$ to be the subset of edges with at least one endpoint  in $B$, and let $E_{S} =  \{ (u,v) \in E : \text{either } u  \in S \text{ or } v \in S \}$ be the subset of edges with at least one endpoint in $S$. We define the subgraphs $G_B = (V, E_B)$ and $G_S = (V, E_S)$.  

\begin{observation}
\label{ob:clarify}
We have $\N(v, E) = \N(v, E_B)$ for all big nodes $v \in B$, and $\N(u, E) = \N(u, E_S)$ for all small nodes $u \in S$.
\end{observation}

\medskip
\noindent {\bf Overview of our approach.} Our algorithm  maintains the following structures.
\begin{itemize}
\item A fractional $b$-matching as per Theorem~\ref{th:sample:b-matching}.
\item  A random subset  $H_B \subseteq E_B$, and a weight function $w^B : H_B \rightarrow [0,1]$ in the subgraph $G_B(H) = (V, H_B)$, as per Definition~\ref{def:H_B}.
\item  A random subset  $H_S \subseteq E_S$, and a weight function $w^S : H_S \rightarrow [0,1]$ in the subgraph $G_S(H) = (V, H_S)$, as per Definition~\ref{def:H_S}. 
\item  A maximal $b$-matching $M_S \subseteq H_S$ in the subgraph $G_S(H)$, that is, for every edge $(u,v) \in H_S \setminus M_S$, there is a node $q \in \{u,v\}$ such that $\text{deg}(q, M_S) = c_q$. 
\item  The set of edges $E^* = \{ e \in E : w(e) = 1\}$.
\end{itemize}
The rest of this section is organized as follows. In Lemmas~\ref{cor:sample:E_B} (resp. Lemma~\ref{cor:sample:E_S}), we prove some properties of the random set $H_B$ (resp. $H_S$) and the weight function $w^B$ (resp. $w^S$). In Lemma~\ref{cor:sample:runtime}, we show that the edge-sets $H_B, H_S, M_S$ and $E^*$ can be maintained in a dynamic setting in $O(\log^3 n)$ amortized update time. In Theorem~\ref{th:sample:main}, we prove our main result, by showing that one of the edge-sets $H_B, M_S, E^*$ is an $O(1)$-approximation to the optimal $b$-matching with high probability.

\medskip
\noindent The proofs of Lemmas~\ref{cor:sample:E_B},~\ref{cor:sample:E_S} and~\ref{cor:sample:runtime} appear in Sections~\ref{sec:cor:sample:E_B},~\ref{sec:cor:sample:E_S} and~\ref{sec:cor:sample:runtime} respectively.

\renewcommand{\E}{\mathbf{E}}

\begin{definition}
\label{def:H_B}
The random set  $H_B \subseteq E_B$ and the weight function $w^B : H_B \rightarrow [0,1]$ are defined so as to fulfill the following conditions.
\begin{eqnarray}
\text{With probability one, we have } \text{deg}(v, H_B) \leq c_v \text{ for every small node } v \in S. \label{eq:w^B:1} \\
 \Pr[e \in H_B] = w(e) \text{ for every edge } e \in E_B. \label{eq:w^B:2} \\
\forall v \in B, \text{ the events } \{ [(u,v) \in H_B] \}, u \in \N(v, E_B), \text{ are mutually independent. }  \label{eq:w^B:3} \\
\text{For each edge } e \in H_B, \text{ we have } w^B(e) =  1   \label{eq:w^B:4}
\end{eqnarray}
We define $Z_B(e) \in \{0, 1\}$ to be an indicator random variable that is set to one if $e \in H_B$ and zero otherwise. 
\end{definition}

\begin{definition}
\label{def:H_S}
The random set  $H_S \subseteq E_S$ and the weight function $w^S : H_S \rightarrow [0,1]$ are defined so as to fulfill the following conditions.
\begin{eqnarray}
\Pr[e \in H_S] = p_e =  \min(1, w(e) \cdot (c \lambda \log n/\epsilon)) \ \ \forall e \in E_S. \label{eq:w^S:1} \\
\text{The  events } \{ [e \in H_S] \}, e \in E_S, \text{ are mutually independent. }  \label{eq:w^S:2} \\
\text{For each edge } e \in H_S,  \text{we have } w^S(e) = \begin{cases} w(e) & \text{ if }  p_e \geq 1; \\
\epsilon/(c \lambda \log n) & \text{ if } p_e < 1. \\
\end{cases} \label{eq:w^S:3} 
\end{eqnarray}
We define $Z_S(e) \in \{0, 1\}$ to be an indicator random variable that is set to one if $e \in H_S$ and zero otherwise. 
\end{definition}

\begin{lemma}
\label{cor:sample:E_B}
For every node $v \in V$, define $W^B_v = \sum_{u \in \N(v, H_B)} w^B(u,v)$. Then the following conditions hold with high probability.
\begin{itemize}
\item For every node $v \in V$, we have $W^B_v \leq c_v$.
\item For every node $v \in B \cap T$, we have $W^B_v \geq (1-\epsilon) \cdot (c_v/\lambda)$.
\end{itemize}
\end{lemma}

\begin{lemma}
\label{cor:sample:E_S}
For every node $v \in V$, define $W^S_v = \sum_{u \in \N(v, H_S)} w^S(u,v)$. The following conditions hold with high probability.
\begin{itemize}
\item For each node $v \in V$, we have $W^S_v \leq c_v$.
\item For each node $v \in S$, we have $\text{deg}(v, H_S) = O(\log^2 n)$.
\item For each node $v \in S \cap T$, we have $W^S_v \geq (1-\epsilon) \cdot (c_v/\lambda)$. 
\end{itemize}
\end{lemma}

\begin{lemma}
\label{cor:sample:runtime}
With high probability, we can maintain the random sets of edges $H_B$ and $H_S$, a maximal $b$-matching $M_S$ in the random subgraph $G_S(H) = (V, H_S)$, and the set of edges $E^*$ in $O(\log^3 n)$-amortized update time.
\end{lemma}

\begin{theorem}
\label{th:sample:main}
With high probability, we can maintain a $O(1)$-approximate $b$-matching in the input graph $G = (V, E)$ in $O(\log^3 n)$ amortized update time.
\end{theorem}

\subsection{Proof of Theorem~\ref{th:sample:main}}
We maintain the random sets of edges $H_B$ and $H_S$, a maximal $b$-matching $M_S$ in the subgraph $G_S(H) = (V, H_S)$, and the set of edges $E^* = \{ e \in E : w(e) = 1\}$ as per Lemma~\ref{cor:sample:runtime}. This requires $O(\log^3 n)$ amortized update time with high probability.  The theorem will follow from Theorem~\ref{th:sample:b-matching}, Lemma~\ref{lm:sample:main:1} and Lemma~\ref{lm:sample:main:3}. 

\begin{lemma}
\label{lm:sample:main:1}
With high probability, each of the edge-sets $H_B, M_S$ and $E^*$ is a valid $b$-matching in $G$. 
\end{lemma}

\begin{proof}
Since $w^B(e) = 1$ for every edge $e \in H_B$ (see Definition~\ref{def:H_B}), Lemma~\ref{cor:sample:E_B} implies that the edge-set $H_B$ is a $b$-matching in $G$ with high probability.

Next, by definition, the edge-set $M_S$ is a $b$-matching in $G_S(H) = (V, H_S)$. Since $H_S \subseteq E$, the edge-set $M_S$ is also a $b$-matching in $G$.

Finally, since $w : E \rightarrow [0,1]$ is a fractional $b$-matching in $G$, the set of edges $E^*$ is also a  $b$-matching in $G$. 
 \
\end{proof}

\begin{lemma}
\label{lm:sample:main:2}
We have $w(E^*) + \sum_{v \in B \cap T} W_v + \sum_{v \in S \cap T} W_v \geq w(E)$.
\end{lemma}

\begin{proof}
Consider any edge $(u,v) \in E$. If $u \notin T$ and $v \notin T$, then by equation~\ref{eq:w:2}, we must have $(u,v) \in E^*$. In contrast, if there is some node $x \in \{u, v \}$ such that $x \in T$, then we must have either $x \in B \cap T$ or $x \in S \cap T$. 

In other words, every edge $(u,v)$ satisfies this property: Either $(u, v) \in E^*$, or it is incident upon some node in $B \cap T$, or it is incident upon some node $S \cap T$. Thus, each edge $e \in E$ contributes at least $w(e)$ to the sum $w(E^*) + \sum_{v \in B \cap T} W_v + \sum_{v \in S \cap T} W_v$. The lemma follows.
 \ \end{proof}

\begin{lemma}
\label{lm:sample:main:3}
We have $w(E) \leq O(1) \cdot \max(|E^*|, |H_B|, |M_S|)$ with high probability.
\end{lemma}

\begin{proof}
Note that $w(E^*) = |E^*|$. We consider three possible cases, based on Lemma~\ref{lm:sample:main:2}.

\medskip
\noindent {\em Case 1.} $w(E^*) \geq (1/3) \cdot w(E)$. In this case, clearly  $w(E) \leq 3 \cdot \max(|E^*|, |H_B|, |M_S|)$.

\medskip
\noindent {\em Case 2.} $\sum_{v \in B \cap T} W_v \geq (1/3) \cdot w(E)$. In this case, we condition on the event under which Lemma~\ref{cor:sample:E_B} holds. Thus, we get:
\begin{eqnarray*}
w(E) & \leq & \sum_{v \in B \cap T} 3 \cdot W_v \leq \sum_{v \in B \cap T} 3 \cdot c_v \leq  \sum_{v \in B \cap T} (3 \lambda /(1-\epsilon)) \cdot W^B_v \\
& \leq & (3 \lambda /(1-\epsilon)) \cdot \sum_{e \in H_B} 2 \cdot w^B(e) = (6\lambda/(1-\epsilon)) \cdot |H_B|
\end{eqnarray*}

\medskip
\noindent {\em Case 3.} $\sum_{v \in S \cap T} W_v \geq (1/3) \cdot w(E)$. In this case, we condition on the event under which Lemma~\ref{cor:sample:E_S} holds. Thus, we get:
\begin{eqnarray*}
w(E) & \leq & \sum_{v \in S \cap T} 3 \cdot W_v \leq \sum_{v \in S \cap T} 3 \cdot c_v \leq  \sum_{v \in S \cap T} (3 \lambda /(1-\epsilon)) \cdot W^S_v \\
& \leq & (3 \lambda /(1-\epsilon)) \cdot \sum_{e \in H_S} 2 \cdot w^S(e) = (6\lambda/(1-\epsilon)) \cdot  \sum_{e \in H_S}  w^S(e) \\
& \leq & (12 \lambda/(1-\epsilon)) \cdot |M_S|.
\end{eqnarray*}
The last inequality holds since $M_S$ is a maximal $b$-matching in $G_S(H) = (V, H_S)$, and since every maximal $b$-matching is a $2$-approximation to the maximum fractional $b$-matching (this follows from LP duality). Accordingly,  we have $\sum_{e \in H_S} w^S(e) \leq 2 \cdot |M_S|$. 
 \ \end{proof}

Since $\lambda, \epsilon$ are constants, this concludes the proof of Theorem~\ref{th:sample:main}.

\subsection{Proof of Lemma~\ref{cor:sample:E_B}}
\label{sec:cor:sample:E_B}

\begin{lemma}
\label{lm:sample:E_B:lowerbound}
With high probability, we have $W^B_v \geq (1-\epsilon) \cdot (c_v/\lambda)$ for every  node $v \in B \cap T$.
\end{lemma}

\begin{proof}
Fix any  node $v \in B \cap T$. Note that $\N(v, E_B) = \N(v, E)$, $W_v \geq c_v/\lambda$, and $c_v \geq c \lambda \log n/\epsilon$. Linearity of expectation, in conjunction with equations~\ref{eq:w^B:2},~\ref{eq:w^B:4} and Observation~\ref{ob:clarify} imply that we have $\E[W^B_v] = \sum_{u \in \N(v, E_B)} \E[Z_B(u,v)] = \sum_{u \in \N(v, E_B)} w(u,v) = \sum_{u \in \N(v, E)} w(u,v) = W_v \geq c_v/\lambda \geq c  \log n/\epsilon$.  Thus, applying Chernoff bound, we infer that $\E[W^B_v] \geq (1-\epsilon) \cdot (c_v/\lambda)$ with high probability. The lemma follows if we take a union bound over all nodes $v \in B \cap T$.
 \ \end{proof}

\begin{lemma}
\label{lm:sample:E_B:upperbound}
With high probability, we have $W^B_v \leq c_v$ for every node $v \in V$.
\end{lemma}

\begin{proof}
Consider any node $v \in V$. If $v \in S$, then we have $W^B_v \leq c_v$ with probability one (see equations~\ref{eq:w^B:1},~\ref{eq:w^B:4}). 

For the rest of the proof, suppose that $v \in B$. Applying an argument similar to the one used in the proof of Lemma~\ref{lm:sample:E_B:lowerbound}, we infer that $\E[W^B_v] = W_v \leq c_v/\gamma$. The last inequality holds due to equation~\ref{eq:w:1}. Since $\gamma > (1+\epsilon)$ and $c_v \geq c \lambda \log n/\epsilon$, applying Chernoff bound we derive that $W^B_v \leq c_v$ with high probability.

Thus,  for each node $v \in V$, we have $W^B_v \leq c_v$ with high probability. The lemma now follows if we take a union bound over all nodes $v \in B$. 
 \ \end{proof}

Lemma~\ref{cor:sample:E_B} now follows from Lemmas~\ref{lm:sample:E_B:lowerbound} and~\ref{lm:sample:E_B:upperbound}.

\subsection{Proof of Lemma~\ref{cor:sample:E_S}}
\label{sec:cor:sample:E_S}

%We define $W_x = \sum_{y \in \N(x, E_S)} w(x,y)$ and $W^*_x = \sum_{y \in \N(x, E_S)} w^S(x,y)$ for every node $x \in V$. 

\subsubsection{High Level Overview}

In order to highlight the main idea, we assume that $p_e < 1$ for every edge $e \in E_S$. First, consider any small node $v \in S$. Since $\N(v, E_S) = \N(v, E)$, from equations~\ref{eq:w:1},~\ref{eq:w^S:1},~\ref{eq:w^S:3} and linearity of expectation, we infer that $\E[\text{deg}(v, H_S)] = (c\lambda \log n/\epsilon) \cdot W_v \leq (c \lambda \log n/\epsilon) \cdot (c_v/(1+\epsilon))$. Since $c_v \in [1, c \log n]$, from equation~\ref{eq:w^S:2} and Chernoff bound we infer that $\text{deg}(v, H_S) \leq (c \lambda \log n/\epsilon) \cdot c_v = O(\log^2 n)$ with high probability. Next, note that $W_v^S = \text{deg}(v, H_S) \cdot (\epsilon/(c\lambda \log n))$. Hence, we also get $W_v^S \leq c_v$ with high probability. Next, suppose that $v \in S \cap T$. In this case, we have $\E[\text{deg}(v, H_S)] = (c\lambda \log n/\epsilon) \cdot W_v \geq (c \lambda \log n/\epsilon) \cdot (c_v/\lambda)$. Again, since this expectation is sufficiently large, applying Chernoff bound we get $\text{deg}(v, H_S) \geq (c \lambda \log n/\epsilon) \cdot (1-\epsilon) \cdot (c_v/\lambda)$ with high probability. It follows that  $W_v^S = (\epsilon/(c \lambda \log n)) \cdot \text{deg}(v, H_S) \geq (1-\epsilon) \cdot (c_v/\lambda)$ with high probability. 

Finally, applying a similar argument we can show that for every big node $v \in B$, we have  $W_v^S \leq c_v$ with high probability.

\subsubsection{Full Details}
For every node $v \in V$, we partition the node-set $\N(v, E_S)$  into two subsets -- $X(v)$ and $Y(v)$ -- as defined below.
\begin{eqnarray}
\label{eq:def:Xv}
X(v) = \{u \in \N(v, E_S) : p_{(u,v)} = 1\} \\
Y(v) = \{u \in \N(v, E_S) : p_{(u,v)} < 1\} \label{eq:def:Yv}
\end{eqnarray}
Next, for every node $v \in V$, we define:
\begin{eqnarray}
\label{eq:def:deltaX}
\delta_X(v) = \sum_{u \in X(v)} w(u,v) \\
\delta_Y(v)  = \sum_{u \in Y(v)} w(u,v) \label{eq:def:deltaY}
\end{eqnarray} 
Since $\N(v, E_S) \subseteq \N(v, E)$ for every node $v \in V$, by equation~\ref{eq:w:1}  we have:
\begin{equation}
\label{eq:sampling:small:1}
\sum_{u \in \N(v, E_S)} w(u,v) = \delta_X(v) + \delta_Y(v) \leq c_v/\gamma
\end{equation}
Since  $X(v) \subseteq \N(v, E_S)$ and $w^S(u,v) = w(u,v)$ for every node $u \in X(v)$, we get:
\begin{equation}
\label{eq:sampling:small:2}
 \sum_{u \in X(v)} w^S(u,v)  = \delta_X(v).
\end{equation}

\begin{lemma}
\label{lm:sampling:small:1}
For every node $v \in V$, if $\delta_Y(v) \leq  \epsilon/\lambda$, then with high probability, we have:
\begin{eqnarray*}
|Y(v) \cap \N(v, H_S)| & \leq & (1+\epsilon) \cdot c \log n; \text{ and } \\
 \sum_{u \in Y(v) \cap \N(v, H_S)} w^S(u,v) & \leq & 2 \epsilon/\lambda.
\end{eqnarray*}
\end{lemma}

\begin{proof}
Recall that for every node $u \in Y(v)$, we have defined $Z_S(u,v) \in \{0,1\}$ to be an indicator random variable that is set to one if $(u,v) \in H_S$ and zero otherwise. Clearly, we have $\E[Z_S(u,v)] = (c \lambda \log n/\epsilon) \cdot w(u,v)$ for all $u \in Y(v)$. Applying linearity of expectation, we get: 
\begin{eqnarray*}
\E\left[|Y(v) \cap \N(v, H_S)|\right] = E\left[\sum_{u \in Y(v)} Z_S(u,v)\right] & = & (c \lambda \log n/\epsilon) \cdot \sum_{u \in Y(v)} w(u,v)  \\
& = & (c \lambda \log n/\epsilon) \cdot \delta_Y(v) \leq  c \log n.
\end{eqnarray*}
Since $\E\left[|Y(v) \cap \N(v, H_S)|\right] \leq  c \log n$, applying Chernoff bound we infer that $|Y(v) \cap \N(v, H_S)| \leq (1+\epsilon)c \log n$ with high probability. 

Finally,  note that each node $u \in Y(v) \cap \N(v, H_S)$ has $w^S(u,v) = \epsilon/(c \lambda \log n)$. This implies that $\sum_{u \in Y(v) \cap \N(v, H_S)} w^S(u,v) = \epsilon/(c \lambda \log n) \cdot |Y(v) \cap H_S|$. Since $|Y(v) \cap H_S| \leq (1+\epsilon) c \log n$ with high probability, we get: $\sum_{u  \in Y(v) \cap \N(v, H_S)} w^S(u,v) \leq  (1+\epsilon)  \epsilon/\lambda \leq 2 \epsilon/\lambda$ with high  probability. This concludes the proof of the lemma.  \
\end{proof}

\begin{lemma}
\label{lm:sampling:small:2}
For every node $v \in V$, if $\delta_Y(v) \geq  \epsilon/\lambda$, then with high probability, we have:
\begin{eqnarray*}
(c \lambda \log n/\epsilon) \cdot \frac{\delta_Y(v)}{(1+\epsilon)} \leq |Y(v) \cap \N(v, E_S)| \leq (c \lambda \log n/\epsilon) \cdot (1+\epsilon) \delta_Y(v); \text{ and } \\
\frac{\delta_Y(v)}{(1+\epsilon)} \leq \sum_{u \in Y(v) \cap \N(v, H_S)} w^S(u,v) \leq (1+\epsilon) \delta_Y(v).
\end{eqnarray*}
\end{lemma}

\begin{proof}
Let $\mu = E[|Y(v) \cap \N(v, H_S)|]$. Applying an argument as in the proof of Lemma~\ref{lm:sampling:small:1}, we get: $\mu = (c \lambda \log n/\epsilon) \cdot \delta_Y(v) \geq c \log n$. Hence, applying Chernoff bound, we infer that $\mu/(1+\epsilon) \leq |Y(v) \cap \N(v, H_S)| \leq (1+\epsilon) \mu$ with high probability. This proves the first part of the lemma.

To prove the second part of the lemma, we simply   note that, as in the proof of Lemma~\ref{lm:sampling:small:1}, we have $\sum_{u \in Y(v) \cap \N(v, H_S)} w^S(u,v) = (\epsilon/(c \lambda \log n)) \cdot |Y(v) \cap \N(v, H_S)|$.  \
\end{proof}

\begin{lemma}
\label{lm:sampling:small:maxdeg}
For every node $v \in V$, we have $\text{deg}(v, H_S) = O\left((\log n /\epsilon) \cdot c_v\right)$ with high probability.
\end{lemma}

\begin{proof}
Fix any node $v \in V$. Note that $X(v) \subseteq \N(v, H_S)$ and $w(u,v) = w^S(u,v) \geq \epsilon/(c \lambda \log n)$ for every node $u \in X(v)$. By equation~\ref{eq:sampling:small:2}, we have $\sum_{u \in X(v)} w^S(u,v) = \delta_X(v)$ for every node $v \in V$. Thus, we get:
\begin{equation}
\label{eq:sampling:small:maxdeg}
|X(v)| \leq (c \lambda \log n/\epsilon) \cdot \delta_X(v) = O\left( (\log n/\epsilon) \cdot \delta_X(v) \right)
\end{equation}
Lemmas~\ref{lm:sampling:small:1} and~\ref{lm:sampling:small:2} imply that with high probability, we have:
\begin{eqnarray}
|Y(v) \cap H_S| & \leq & \max\left(c\log n,  (c \lambda \log n/\epsilon) (1+\epsilon) \delta_Y(v)\right) \nonumber \\
& = & O\left((\log n/\epsilon) \cdot \delta_Y(v)\right) \label{eq:sampling:small:maxdeg:1}
\end{eqnarray}
Since $\text{deg}(v, H_S) = |X(v)| + |Y(v) \cap \N(v, H_S)|$, the lemma follows if we add equations~\ref{eq:sampling:small:maxdeg} and~\ref{eq:sampling:small:maxdeg:1}, and recall that $\delta_X(v) + \delta_Y(v) \leq c_v$ (see equation~\ref{eq:sampling:small:1}).
 \ \end{proof}

\begin{lemma}
\label{lm:sampling:small:weight}
For every node $v \in V$, we have $W^S_v \leq  c_v$ with high probability. 
\end{lemma}

\begin{proof}
Lemmas~\ref{lm:sampling:small:1} and~\ref{lm:sampling:small:2} imply that with high probability, we have:
\begin{eqnarray}
\label{eq:lm:sampling:small:weight:1}
\sum_{u \in Y(v) \cap \N(v, H_S)} w^S(u,v) \leq \max\left(2\epsilon/\lambda , (1+\epsilon) \delta_Y(v)\right) 
\end{eqnarray}
Since the node-set $\N(v, H_S)$ is partitioned into $X(v)$ and $Y(v) \cap \N(v, H_S)$,  we get:
\begin{eqnarray}
W^S_v & = & \sum_{u \in X(v)} w^S(u,v) + \sum_{u \in Y(v) \cap \N(v, H_S)} w^S(u,v)  \nonumber  \\
& \leq & (1+\epsilon) \cdot \delta_X(v) + \max(2 \epsilon/\lambda, (1+\epsilon) \delta_Y(v)) \label{eq:new:1} \\
& \leq & (1+\epsilon) \cdot (\delta_X(v) + \delta_Y(v)) + 2 \epsilon/\lambda \nonumber \\
& \leq & (1+\epsilon) \cdot (c_v/\gamma) + (2 \epsilon/\lambda) \cdot c_v \label{eq:new:2} \\
& \leq & (1+\epsilon) \cdot (c_v/\gamma) + 2\epsilon \cdot (c_v/\gamma) \label{eq:new:3} \\
& \leq & c_v \label{eq:new:4}
\end{eqnarray}
Equation~\ref{eq:new:1} follows from equations~\ref{eq:sampling:small:2} and~\ref{eq:lm:sampling:small:weight:1}, and it holds with high probability.  Equation~\ref{eq:new:2} follows from equation~\ref{eq:sampling:small:1} and the fact that $c_v \geq 1$. Equation~\ref{eq:new:3} holds since $\gamma < \lambda$ (see Theorem~\ref{th:sample:b-matching}). Equation~\ref{eq:new:4} holds since $\gamma > 1+3\epsilon$  (see Theorem~\ref{th:sample:b-matching}).
 \ \end{proof}

\begin{lemma}
\label{lm:sampling:small:matching}
For every node $v \in S \cap T$, we have  $W^S_v \geq (1-\epsilon) \cdot \left(c_v/\lambda\right)$.
\end{lemma}

\begin{proof}
Fix any node $v \in S \cap T$. Since $v \in S$, we have $\N(v, E) = \N(v, E_S)$. Since $v \in T$, we have $W_v = \sum_{u \in \N(v, E_S)} w(u,v)  \geq c_v/\lambda$. Since $\sum_{u \in \N(v, E_S)} w(u,v) = \delta_X(v) + \delta_Y(v)$, we get:
\begin{equation}
\label{eq:sampling:small:matching:1}
\delta_X(v) + \delta_Y(v) \geq c_v/\lambda
\end{equation}
We also recall that by equation~\ref{eq:sampling:small:2} we have:
\begin{equation}
\label{eq:sampling:small:matching:2}
\sum_{u \in X(v)} w^S(u,v)  = \delta_X(v)
\end{equation}

We now consider two possible cases, based on the value of  $\delta_Y(v)$.

\medskip
\noindent {\em Case 1.} We have $\delta_Y(v) \leq \epsilon  /\lambda$. Since $c_v \geq 1$, in this case, we have $\delta_X(v) \geq c_v/\lambda - \delta_Y(v) \geq c_v(1-\epsilon)/\lambda$. By equation~\ref{eq:sampling:small:matching:2}, we infer that $W^S_v \geq \sum_{u \in X(v)} w^S(u,v) = \delta_X(v) \geq c_v(1-\epsilon)/\lambda$. This concludes the proof of the lemma for Case 1.

\medskip
\noindent {\em Case 2.} We have $\delta_Y(v) > \epsilon /\lambda$. In this case,  Lemma~\ref{lm:sampling:small:2} implies that with high probability we have: $\sum_{u \in Y(v) \cap \N(v, H_S)} w^S(u,v) \geq \delta_Y(v)/(1+\epsilon)$. Since the node-set $\N(v, H_S)$ is partitioned into $X(v)$ and $Y(v) \cap \N(v, H_S)$, we get:
\begin{eqnarray*}
W^S(u,v) = \sum_{u \in X(v)} w^S(u,v) + \sum_{u \in Y(v) \cap \N(v, H_S)} w^S(u,v) \geq \delta_X(v) + \delta_Y(v)/(1+\epsilon) \\
\geq (\delta_X(v) + \delta_Y(v))/(1+\epsilon) \geq (c_v/\lambda) \cdot (1/(1+\epsilon)) \geq (1-\epsilon) \cdot (c_v/\lambda)
\end{eqnarray*}
This concludes the proof of the lemma for Case 2.
 \ \end{proof}

\noindent Lemma~\ref{cor:sample:E_S} follows from Lemmas~\ref{lm:sampling:small:maxdeg},~\ref{lm:sampling:small:weight},~\ref{lm:sampling:small:matching}, and the fact that $c_v = O(\log n)$ for all $v \in S$.

\subsection{Proof of Lemma~\ref{cor:sample:runtime}}
\label{sec:cor:sample:runtime}

We maintain the fractional $b$-matching $\{w(e)\}$ as per Theorem~\ref{th:sample:b-matching}. This requires $O(\log n)$ amortized update time, and  starting from an empty graph, $t$ edge insertions/deletions in $G$ lead to $O(t \log n)$ many changes in the edge-weights $\{w(e)\}$. Thus, we can easily  maintain the edge-set $E^* = \{e \in E : w(e) = 1\}$ in $O(\log n)$ amortized update time. Specifically, we store the edge-set $E^*$ as a doubly linked list. For every edge $(u,v) \in E^*$, we maintain a pointer that points to the position of $(u,v)$ in this linked list. For every edge $(u,v) \in E \setminus E^*$, the corresponding pointer is set to NULL.  An edge $(u,v)$ is inserted into/deleted from the set $E^*$ only when its weight $w(e)$ is changed. Thus, maintaining the linked list for $E^*$ does not incur any additional overhead in the update time.

Next, we show to maintain the edge-set $H_S$ by independently sampling each edge $e \in E_S$ with probability $p_e$. This probability is completely determined by the weight $w(e)$. So we need to resample the edge each time its weight changes. Thus, the amortized update time for maintaining $H_S$  is also $O(\log n)$.  Similar to the case of the edge-set $E^*$, we store the edge-set $H_S$ as a doubly linked list.

Next, we show how to maintain the maximal $b$-matching $M_S$ in $H_S$. Every edge $e \in H_S$ has at least one endpoint in $S$, and each node $v \in S$ has $\text{deg}(v, H_S) = O(\log^2 n)$ with high probability (see Lemma~\ref{cor:sample:E_S}).  Exploiting this fact, for each node $v \in B$, we can maintain the set of its free (unmatched) neighbors $F_v(S) = \{ u \in \N(v, H_S) : u \text{ is unmatched in } M_S \}$ in $O(\log^2 n)$ amortized time per update in $H_S$, with high probability. This is done as follows. Since $v \in B$, the onus of maintaining the set $F_v(S)$ falls squarely upon the nodes in $\N(v, H_S) \subseteq S$. Specifically, each small node $u \in S$ maintains a ``status-bit'' indicating if it is free or not. Whenever a matched small node $u$ changes its status-bit, it communicates this information to its  neighbors in $\N(u, H_S) \cap B$ in $O(\text{deg}(u, H_S)) = O(\log^2 n)$ time. Using the lists $\{F_v(S)\}, v \in B,$ and the status-bits of the small nodes, after each edge insertion/deletion in $H_S$, we can update the maximal $b$-matching $M_S$ in $O(\log^2 n)$ worst case time, with high probability. Since each edge insertion/deletion in $G$, on average, leads to $O(\log n)$ edge insertions/deletions in $H_S$, we spend $O(\log^3 n)$ amortized update time, with high probability, for maintaining the matching $M_S$.

Finally,  we show how to maintain the set $H_B$. The edges $(u,v) \in E_B$ with both endpoints $u, v \in B$ are sampled independently with probability $w(u,v)$. This requires $O(\log n)$ amortized update time. Next, each small node $v \in S$ randomly selects some neighbors $u \in \N(v, E_B)$ and adds the corresponding edges $(u,v)$ to the set $H_B$, ensuring that $\Pr[(u,v) \in H_B] = w(u,v)$ for all $u \in \N(v, E_B)$ and that $\text{deg}(v, H_B) \leq c_v$. The random choices made by the different small nodes are mutually independent, which implies  equation~\ref{eq:w^B:3}. But, for a given  node $v \in S$ the random variables $\{Z_B(u,v)\}, u \in \N(v, E_B),$ are completely correlated.  They are determined as follows. 

In the beginning, we pick a number $\eta_v$ uniformly at random from the interval $[0,1)$, and, in a predefined manner, label the set of big nodes as $B = \{v_1, \ldots, v_{|B|}\}$. For each $i \in \{1, \ldots, |B|\}$, we define $a_i(v) = w(v, v_i)$ if $v_i \in \N(v, E_B)$ and zero otherwise. We also define $A_i(v) = \sum_{j=1}^i a_{j}(v)$ for each $i \in \{1, \ldots, |B|\}$ and set $A_0(v) = 0$. At any given point in time, we define $\N(v, H_B) = \{v_i \in B : A_{i-1}(v) \leq k + \eta_v < A_i(v) \text{ for some nonnegative integer } k < c_v\}$. Under this scheme, for every node $v_i \in B$, we have  $\Pr[v_i \in \N(v, H_B)] = A_i(v) - A_{i-1}(v) = a_i(v)$. Thus, we get   $\Pr[v_i \in \N(v, H_B)]  = w(v, v_i)$ for all $v_i \in \N(v, E_B)$, and $\Pr[v_i \in \N(v, H_B)] = 0$ for all $v_i \neq \N(v, E_B)$. Also note that  $\text{deg}(v, H_B) \leq  \lceil \sum_{v_i \in \N(v, E_B)} w(v, v_i) \rceil \leq \lceil W_v \rceil \leq \lceil c_v/(\gamma) \rceil \leq c_v$. Hence, equations~\ref{eq:w^B:1},~\ref{eq:w^B:2} are satisfied. We maintain the sums $\{A_i(v)\}, i,$ and the set $\N(v, H_B)$ using a balanced binary  tree data structure, as described below.

We store the ordered sequence of $|B|$ numbers $a_1(v), \ldots, a_{|B|}(v)$  in the leaves of a
balanced binary tree from left to right. Let $x_i$ denote the leaf node that stores the value $a_i(v)$. Further, at each internal node $x$ of the balanced binary tree, we store the sum $S_x = \sum_{i : x_i \in T(x)} a_i(v)$, where $T(x)$ denotes the set of nodes in the subtree rooted at $x$.  This data structure can support the following operations. 

INCREMENT$(i, \delta)$: This asks us to set $a_i(v) \leftarrow a_i(v) + \delta$, where $\delta$ is any real number. To perform this update, we first change the value stored at the leaf node $x_i$. Then starting from the node $x_i$, we traverse up to the root of the tree. At each internal node $x$ in this path from $x_i$ to the root, we set $S_x \leftarrow S_x + \delta$. The $S_x$ values at every other internal node remains unchanged. Since the tree has depth $O(\log n)$, the total time required to update the data structure is also $O(\log n)$. 

RETURN-INDEX$(y)$: Given a number $0 \leq y < c_v$, this asks us to return an index $i$ (if it exists) such that $A_{i-1}(v) \leq y < A_i(v)$. We can answer this query in $O(\log n)$ time by doing binary search. Specifically, we perform the following operations. We initialize a counter $C \leftarrow 0$ and start our binary search at the root of the tree. At an intermediate stage of the binary search, we are at some internal node $x$ and we know that $y < C + S_x$. Let $x(l)$ and $x(r)$ respectively be the left and right child of $x$. Note that $S_x = S_{x(l)} + S_{x(r)}$. If $y < C + S_{x(l)}$, then we move to the node $x(l)$. Otherwise, we set $C \leftarrow C + S_{x(l)}$ and move to the node $x(r)$. We continue this process until we reach a leaf node, which gives us the required answer. The total time taken by the procedure is $O(\log n)$.

We use the above data structure to maintain the sets $\N(v, H_B), v \in S$. Whenever the weight of an edge $(u,v)$, $v \in S$, changes, we can update the set $\N(v, H_B)$ by making one call to the INCREMENT$(i, \delta)$, and  $c_v$ calls to RETURN-INDEX$(y)$, one for each $y = k + \eta_v$, where $k < c_v$ is a nonnegative integer. Since $c_v = O(\log n)$, the total time required is $O(\log^2 n)$ per change in the edge-weights $\{w(e)\}$.

Since each edge insertion/deletion in $G$, on average, leads to $O(\log n)$ changes in the edge-weights $\{w(e)\}$, the overall amortized update time for maintaining the edge-set $H_B$ is $O(\log^3 n)$. 

Similar to the edge-sets $E^*$ and $H_S$, we store the edge-set $H_B$ as a doubly linked list. Each edge $(u,v) \in H_B$ maintains a pointer to its position in this list. Each edge $(u,v) \in E \setminus H_B$ sets the corresponding pointer to NULL.  It is easy to check that this does not incur any additional overhead in the update time.
This  concludes the proof of the lemma.

\section{Conclusion and Open Problems}
In this paper, we introduced a dynamic version of the primal-dual method.  Applying this framework, we obtained the  first nontrivial dynamic algorithms for the set cover and $b$-matching problems. Specifically, we presented a dynamic algorithm for set cover that maintains a $O(f^2)$-approximation in $O(f \cdot \log (m+n))$ update time, where $f$ is the maximum frequency of an element, $m$ is the number of sets and $n$ is the number of elements. On the other hand, for the $b$-matching problem, we presented a dynamic algorithm that maintains a $O(1)$-approximation in $O(\log^3 n)$ update time. Our work leaves several interesting open questions. We conclude the paper by stating a couple of such problems.
\begin{itemize}
\item Recall that in the static setting the set cover problem admits $O(\min (f, \log n))$-approximation in $O(f \cdot (m+n))$-time. Can we match this approximation guarantee in the dynamic setting in $O(f \cdot \text{poly} \log (m+n))$ update time? As a first step, it will be interesting to design a dynamic algorithm for fractional hypergraph $b$-matching that maintains a $O(f)$-approximation and has an update time of $O(f \cdot \text{poly} \log (m+n))$.
\item Are there other well known problems (such as facility location, Steiner tree etc.) that can be solved in the dynamic setting using the primal-dual framework?
\end{itemize}

\bibliographystyle{abbrv}
\bibliography{citations}

\begin{thebibliography}{10}

\bibitem{bmatching}
K.~J. Ahn and S.~Guha.
\newblock Near linear time approximation schemes for uncapacitated and
  capacitated b-matching problems in nonbipartite graphs.
\newblock In {\em Proceedings of the Twenty-Fifth Annual {ACM-SIAM} Symposium
  on Discrete Algorithms, {SODA} 2014, Portland, Oregon, USA, January 5-7,
  2014}, pages 239--258, 2014.

\bibitem{BYE81}
R.~Bar-Yehuda and S.~Even.
\newblock A linear time approximation algorithm for the weighted vertex cover
  problem.
\newblock {\em Journal of Algorithms}, 2:198--203, 1981.

\bibitem{BaswanaGS11}
S.~Baswana, M.~Gupta, and S.~Sen.
\newblock Fully dynamic maximal matching in ${O}(\log n)$ update time.
\newblock In {\em 52nd IEEE Symposium on Foundations of Computer Science},
  pages 383--392, 2011.

\bibitem{BHI15}
S.~Bhattacharya, M.~Henzinger, and G.~F. Italiano.
\newblock Deterministic fully dynamic data structures for vertex cover and
  matching.
\newblock In {\em Procs. 26th Annual {ACM-SIAM} Symposium on Discrete
  Algorithms ({SODA} 2015)}, pages 785--804, 2015.

\bibitem{BuchbinderN09}
N.~Buchbinder and J.~Naor.
\newblock The design of competitive online algorithms via a primal-dual
  approach.
\newblock {\em Foundations and Trends in Theoretical Computer Science},
  3(2-3):93--263, 2009.

\bibitem{DFF56}
G.~B. Dantzig, L.~R. Ford, and D.~R. Fulkerson.
\newblock A primal-dual algorithm for linear programs.
\newblock In H.~W. Kuhn and A.~W. Tucker, editors, {\em Linear Inequalities and
  Related Systems}, pages 171--181. Princeton University Press, Princeton, NJ,
  1956.

\bibitem{EGI09}
D.~Eppstein, Z.~Galil, and G.~F. Italiano.
\newblock Dynamic graph algorithms.
\newblock In M.~J. Atallah and M.~Blanton, editors, {\em Algorithms and Theory
  of Computation Handbook, 2nd Edition, Vol.~1}, pages 9.1--9.28. CRC Press,
  2009.

\bibitem{Feige-setcover}
U.~Feige.
\newblock A threshold of $\text{ln } n$ for approximating set cover.
\newblock {\em Journal of the ACM}, 45:634--652, 1998.

\bibitem{Gabow}
H.~N. Gabow.
\newblock An efficient reduction technique for degree-constrained subgraph and
  bidirected network flow problems.
\newblock In {\em Proceedings of the 15th Annual {ACM} Symposium on Theory of
  Computing, 25-27 April, 1983, Boston, Massachusetts, {USA}}, pages 448--456,
  1983.

\bibitem{GoemansW92}
M.~Goemans and D.~P. Williamson.
\newblock A general approximation technique for constrained forest problems.
\newblock {\em SIAM J. Comput.}, 24:296--317, 1992.

\bibitem{GW97}
M.~X. Goemans and D.~P. Williamson.
\newblock The primal-dual method for approximation algorithms and its
  application to network design problems.
\newblock In D.~S. Hochbaum, editor, {\em Approximation algorithms for NP-hard
  problems}, pages 144--191. PWS Publishing Company, 1997.

\bibitem{GuptaP13}
M.~Gupta and R.~Peng.
\newblock Fully dynamic $(1+\epsilon)$-approximate matchings.
\newblock In {\em 54th IEEE Symposium on Foundations of Computer Science},
  pages 548--557, 2013.

\bibitem{johnson}
D.~S. Johnson.
\newblock Approximation algorithms for combinatorial problems.
\newblock {\em Journal of Computer and System Sciences}, 9:256--278, 1974.

\bibitem{unique-games}
S.~Khot and O.~Regev.
\newblock Vertex cover might be hard to approximate to within $2-\epsilon$.
\newblock {\em Journal of Computer and System Sciences}, 74, 2008.

\bibitem{Korman}
S.~Korman.
\newblock {\em On the Use of Randomization in the Online Set Cover Problem}.
\newblock Weizmann Institute of Science, 2004.

\bibitem{Kuh55}
H.~W. Kuhn.
\newblock The {H}ungarian method for the assignment problem.
\newblock {\em Naval Research Logistics Quarterly}, 2:83--97, 1955.

\bibitem{NeimanS13}
O.~Neiman and S.~Solomon.
\newblock Simple deterministic algorithms for fully dynamic maximal matching.
\newblock In {\em 45th ACM Symposium on Theory of Computing}, pages 745--754,
  2013.

\bibitem{OnakR10}
K.~Onak and R.~Rubinfeld.
\newblock Maintaining a large matching and a small vertex cover.
\newblock In {\em 42nd ACM Symposium on Theory of Computing}, pages 457--464,
  2010.

\bibitem{Vazirani01}
V.~V. Vazirani.
\newblock {\em Approximation Algorithms}.
\newblock Springer-Verlag, New York, NY, USA, 2001.

\end{thebibliography}

\end{document}